\documentclass[journal,12pt,onecolumn,draftclsnofoot,]{IEEEtran}
\usepackage{amsthm,amsmath,amssymb}

\usepackage{mathrsfs}
\usepackage{algorithmic,algorithm}
\usepackage[pdftex]{graphicx}
\usepackage{subfigure} 
\usepackage{cleveref}
\usepackage{cite} 

\begin{document}
\title{Temporal-Structure-Assisted Gradient Aggregation for Over-the-Air Federated Edge Learning}

\author{Dian~Fan, 
        Xiaojun~Yuan,~\IEEEmembership{Senior~Member,~IEEE,}
        and~Ying-Jun~Angela~Zhang,~\IEEEmembership{Fellow,~IEEE}
\thanks{D. Fan and X. Yuan are with the Center for intelligent Networking and Communications, the University of Electronic Science and Technology of China, Chengdu, China (e-mail: df@std.uestc.edu.cn; xjyuan@uestc.edu.cn). Y.-J. A. Zhang is with the Department of Information Engeneering, The Chinese University of Hong Kong, Shatin, New Territories, Hong Kong SAR (e-mail: yjzhang@ie.cuhk.edu.hk). The corresponding author is Xiaojun Yuan.}}

\maketitle

\begin{abstract}
In this paper, we investigate over-the-air model aggregation in a federated edge learning (FEEL) system. We introduce a Markovian probability model to characterize the intrinsic temporal structure of the model aggregation series. With this temporal probability model, we formulate the model aggregation problem as to infer the desired aggregated update given all the past observations from a Bayesian perspective. We develop a message passing based algorithm, termed temporal-structure-assisted gradient aggregation (TSA-GA), to fulfil this estimation task with low complexity and near-optimal performance. We further establish the state evolution (SE) analysis to characterize the behaviour of the proposed TSA-GA algorithm, and derive an explicit bound of the expected loss reduction of the FEEL system under certain standard regularity conditions. In addition, we develop an expectation maximization (EM) strategy to learn the unknown parameters in the Markovian model. We show that the proposed TSA-GA algorithm significantly outperforms the state-of-the-art, and is able to achieve comparable learning performance as the error-free benchmark in terms of both convergence rate and final test accuracy.
\end{abstract}


\IEEEpeerreviewmaketitle

\section{Introduction}
As the fast development of wireless big data, massive amounts of mobile data generated at edge devices with growing computation power have boosted the desire to train artificial intelligence models at the wireless edge. Federated learning (FL) \cite{mcmahan2017communication} is one of the most promising enabling technologies for distributed model training and inference, where a global model is shared and trained collaboratively among local devices using local datasets of their own. The local updates are aggregated at a remote parameter server (PS) which tracks and broadcasts the global model update with the participating devices throughout the training process. Unlike centralized learning which requires direct uploading of raw data, the FL paradigm only involves the uploading of model updates by each individual device, thereby relieving the communication cost significantly and avoiding the exposure of local data.\\
\indent In spite of the appealing aspects of FL compared with centralized learning, it has been reported that the demanding uplink communication overhead of high-dimensional model updating over an unreliable wireless medium turns out to be a critical bottleneck for the implementation of FL \cite{mcmahan2017communication,konevcny2016federated}. Recently, much research effort has been devoted to incorporate the physical layer characteristics into the FL system via the communication-learning joint design, referred to as federated edge learning (FEEL) \cite{nishio2019client,ren2020scheduling,jeon2020compressive,chen2020joint,zhu2019broadband, amiri2020machine, amiri2020federated, yang2020federated,zhu2020one, elgabli2020harnessing,sery2020analog,liu2020reconfigurable,amiri2020blind}. In this thread, over-the-air computation based FEEL has emerged by leveraging the waveform-superposition property of the wireless medium for simultaneous uploading of local updates, leading to a high spectral efficiency compared with conventional orthogonal multiple access protocols \cite{zhu2019broadband, amiri2020machine, amiri2020federated, amiri2020blind, yang2020federated,zhu2020one, elgabli2020harnessing,sery2020analog,liu2020reconfigurable}. Pioneering works in \cite{zhu2019broadband, amiri2020machine, amiri2020federated, yang2020federated} validate the superiority of this over-the-air transmission scheme over the conventional orthogonal one by significant acceleration of convergence. Various over-the-air FEEL approaches \cite{amiri2020blind, elgabli2020harnessing, sery2020analog, zhu2020one,liu2020reconfigurable} have been developed to overcome the hostile effects of wireless links between edge devices and the PS. \\
\indent The intrinsic sparsity of local updates can be leveraged to relieve the bandwidth limitation and improve the learning efficiency of FEEL. This is motivated by the observation that the number of significant elements in a model update is extremely small. Specifically,  \cite{amiri2020machine} proposed to sparsify and compress local updates before transmission. The desired aggregated update at PS is then reconstructed from the noisy received signal via compressed sensing. In \cite{amiri2020federated}, the scheme of \cite{amiri2020machine} is extended to a fading channel, where a truncated channel inversion strategy is employed to confront fading. The existing works \cite{amiri2020machine} and \cite{amiri2020federated}, however, have a common limitation, i.e., they use the sparsity structure of model updates within a single communication round but ignore the more obscure structure of the updates between rounds. It is known in the artificial intelligence community that the significant model parameters are highly correlated throughout the training process; see, e.g., the work on model pruning \cite{han2015learning}. This inspires us to explore the intrinsic temporal correlation of model updates as a new dimension to enhance the FEEL performance.\\
\indent In this paper, we make an initial attempt to investigate the temporal structure of gradients as local updates in the FEEL system with over-the-air update aggregation. We  consider a FEEL system with over-the-air computation over a multiple access channel (MAC) where local gradients are sparsified and compressed before transmission to meet the bandwidth limitation. We introduce a probability model to characterize the intrinsic temporal structure of the gradient aggregation series, i.e., the strong temporal correlation is characterized by two independent Markov chains for each element, one for support and the other for amplitude. With this probabilistic model as prior information, the goal of the PS is formulated as to infer the minimum mean-squared error (MMSE) solution of the desired aggregated update given all the past observations in an online fashion. We develop a message passing based algorithm, termed temporal structure assisted gradient aggregation (TSA-GA), to approximately fulfil this estimation task with relatively low complexity. Compared with the state-of-the-art work in \cite{amiri2020machine} which only exploits the intra-round sparsity of gradient, our reconstruction approach benefits from additional prior knowledge of the inter-round gradient correlation. In addition, we emphasize that the TSA-GA algorithm is an extension of the turbo compressed sensing (Turbo-CS) algorithm \cite{ma2014turbo} for online sparse signal recovery. The Turbo-CS framework generally has the advantage of faster convergence and better recovery performance, as compared with the approximate message passing (AMP) based compressed sensing approach \cite{ma2015performance}.\\
\indent Furthermore, we establish the state evolution (SE) analysis to characterize the behaviour of the proposed TSA-GA algorithm as well as the convergence of the over-the-air FEEL system. Specifically, we extend the SE in \cite{ma2014turbo} to analyse our online TSA-GA algorithm by tracking two scalar state variables recursively in each communication round until a fixed point is achieved. We then prove the monotonicity of the SE fixed point sequences over communication rounds. Based on that, by imposing standard assumptions on the FEEL loss function \cite{friedlander2012hybrid}, we establish an explicit bound of the expected loss reduction of the FEEL system. \\
\indent Besides, since model parameters are typically unknown in practical implementation, we develop an expectation maximization (EM) based strategy for determining the unknown parameters in the prior Markovian models. Numerical results confirm that the proposed TSA-GA scheme for over-the-air FEEL with local gradient compression outperforms the state-of-the-art \cite{amiri2020machine} and achieves comparable learning performance with the error-free benchmark in terms of both convergence rate and final test accuracy.\\
\indent The remainder of this paper is organized as follows. In Section II, we describe the FEEL framework, the wireless uplink model and the compression strategy at the edge. In Section III, we introduce the temporal structure of the gradient aggregation and formulate the aggregation reconstruction at the PS as an online Bayesian inference problem. In Section IV, we describe the TSA-GA algorithm to solve this problem approximately with low complexity. In Section V, we present the SE and the convergence analysis for the proposed training scheme. In Section VI, we tackle the practical parameter decision problem based on EM. In Section VII, simulation results are given and the paper concludes in Section VIII.

\section{System Model}
\subsection{Federated Edge Learning System}
\begin{figure}[!t]
\centering
\includegraphics [scale=0.40] {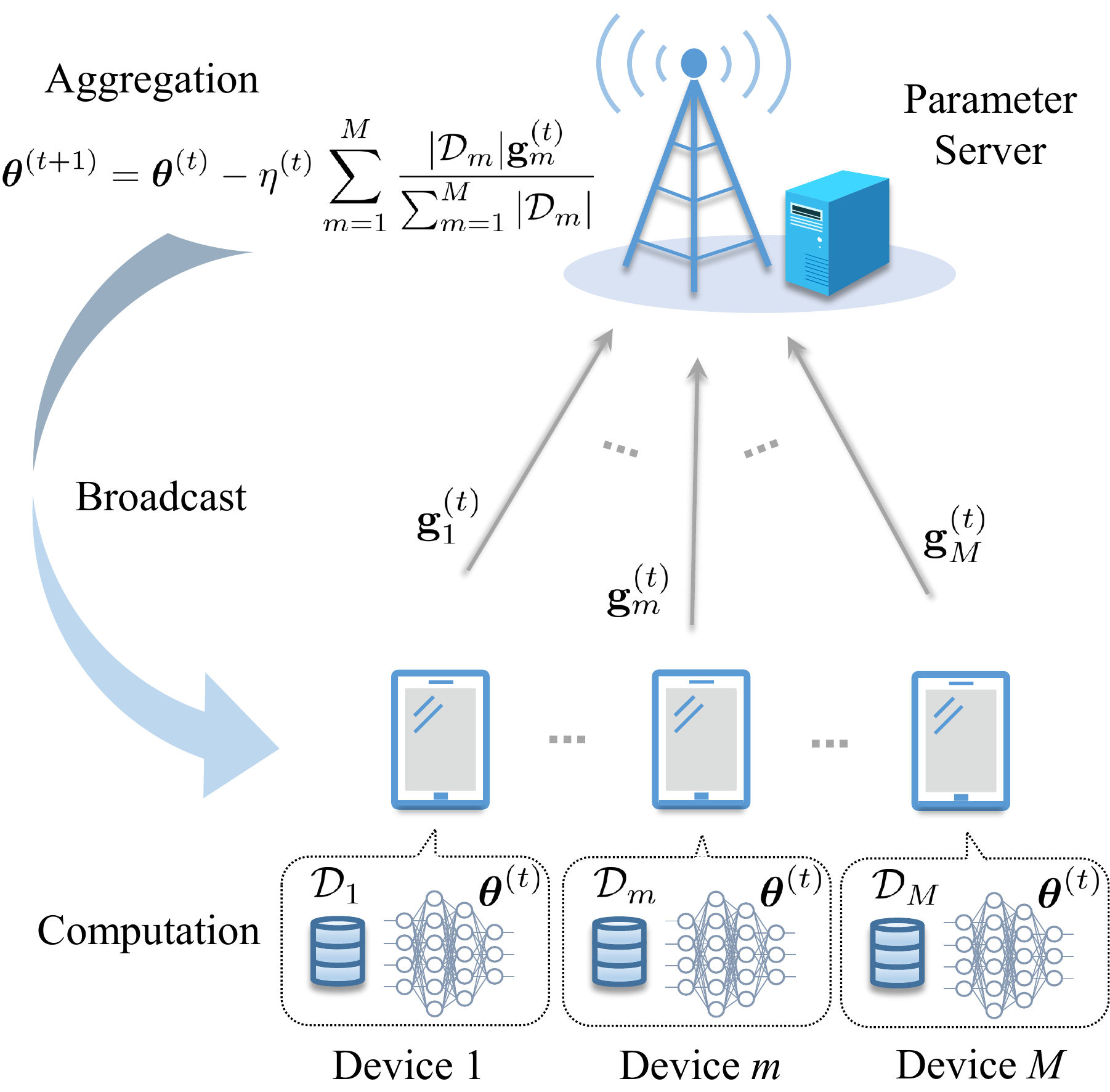}
\caption{An illustration of over-the-air FEEL architecture}
\end{figure}
We consider a FEEL system with a PS sharing the same global model $\boldsymbol{\theta} \in \mathbb{R}^N $ with $M$ edge devices, as illustrated in Fig. 1. Each device $m$ only has access to the local dataset $\mathcal{D}_m=\{(\mathbf{u}_{mk})\}_{k=1}^{K_m}$ of its own, where $K_m=|\mathcal{D}_m|$ is the number of local training samples, and $\mathbf{u}_{mk}$ denotes the $k$-th local training sample of device $m$. Given the sample-wise loss $\ell(\boldsymbol{\theta};\mathbf{u}_{mk})$ specified by the learning objective, the local loss function of device $m$ is written as
\begin{equation}
    \mathcal{L}_m(\boldsymbol{\theta})=\frac{1}{K_m}\sum_{k=1}^{K_m}\ell(\boldsymbol{\theta};\mathbf{u}_{mk}).
\end{equation}
The FEEL task is to minimize the global loss function
\begin{equation}
    \mathcal{L}(\boldsymbol{\theta})=\frac{1}{K}\sum_{m=1}^{M}K_m \mathcal{L}_m(\boldsymbol{\theta}),
\end{equation}
where $K\triangleq\sum_{m=1}^M K_m$ is the total number of data samples.
 This optimization is carried out via local gradient descent on the edge devices. At communication round $t$, starting from the latest model $\boldsymbol{\theta}_m^{(t)}[1]\triangleq \boldsymbol{\theta}^{(t)}$, each device $m$ performs $E$ times of local gradient descent via
\begin{equation}
    \boldsymbol{\theta}_m^{(t)}[i+1]=\boldsymbol{\theta}_m^{(t)}[i]-\eta^{(t)} \nabla \mathcal{L}_m(\boldsymbol{\theta}_m^{(t)}[i]) \label{DSGD},
\end{equation}
for $i=1,...,E$, where $\eta^{(t)}$ is the learning rate of round $t$. Thereafter, the local model update for round $t$ at device $m$ is given by
\begin{equation}
   \mathbf{g}_m^{(t)} = \boldsymbol{\theta}_m^{(t)}[E+1]-\boldsymbol{\theta}_m^{(t)}[1]. \label{4}
\end{equation}
Each device sends $\mathbf{g}_m^{(t)}$ to the PS for the model aggregation according to
\begin{equation}
    \boldsymbol{\theta}^{(t+1)}=\boldsymbol{\theta}^{(t)}-\eta^{(t)} \frac{1}{K} \sum_{m=1}^M K_m \mathbf{g}_m^{(t)} \label{Aggregation}.
\end{equation}
This updated model is shared among all the devices by the PS for the next round of local gradient descent. The training paradigm of (3)-(5) iterates until convergence.

\subsection{Wireless Transmission Model}
The FEEL system assumes a Guassian MAC with $s$ sub-channels as
\begin{equation}
    \tilde{\mathbf{y}}^{(t)}=\sum_{m=1}^{M} \tilde{\mathbf{g}}_m^{(t)}+\tilde{\mathbf{n}}^{(t)} \label{MAC},
\end{equation}
where $\tilde{\mathbf{g}}_m^{(t)}$ is the channel input sent by device $m$ at round $t$, $\tilde{\mathbf{y}}^{(t)} \in \mathbb{R}^s$ is the received signal at the PS, and $\tilde{\mathbf{n}}^{(t)} \in \mathbb{R}^s$ is the additive white Gaussian noise (AWGN) whose elements are independently and identically distributed (IID) according to $\mathcal{N}(0, \sigma_e^2)$. In general, $\tilde{\mathbf{g}}_m^{(t)} \triangleq\Phi(\mathbf{g}_m^{(t)}) \in \mathbb{R}^s$ is a function of the local gradient modified by error-accumulation, sparsification and compression, as specified later in the next subsection.
During a total of $T$ communication rounds, the transmission of each device $m$ is subject to the time-averaged power constraint: 
\begin{equation}
    \frac{1}{T} \sum_{t=1}^{T} \Vert \tilde{\mathbf{g}}_m^{(t)} \Vert^2 \le \bar{P} \label{power constraint}.
\end{equation}
Furthermore, the downlink model broadcasting from the PS to the edge devices is assumed to be error-free by following the convention \cite{amiri2020machine,amiri2020federated}.

\subsection{Gradient Sparsification and Compression}
Directly transmitting high-dimensional local gradients is not desirable due to limited bandwidth and power resource. For communication-efficient FEEL implementation, the model gradient vectors need to be sparsified and compressed before transmission. Since these operations incur error inevitably, we take an error-accumulation strategy to compensate the damage to model update \cite{amiri2020machine}. At each round, the local gradient $\mathbf{g}_m^{(t)}$ is pre-processed by adding the sparsification error $\Delta_m^{(t)}$ accumulated in the previous rounds, i.e.,
\begin{equation}
    {\mathbf{g}_m^{ec^{(t)}}}=\mathbf{g}_m^{(t)}+\Delta_m^{(t)} \label{ec},
\end{equation}
where the initial error is $\Delta_m^{(0)}=0$.
This error-accumulated local gradient is then sparsified by setting all to zero but the $k$ elements with the largest absolute values, denoted as a mapping:
\begin{equation}
    {\mathbf{g}_m^{sp^{(t)}}}=\mathrm{sp}_k({\mathbf{g}_m^{ec^{(t)}}}). \label{spar}
\end{equation}
Accordingly, the local error is updated as
\begin{equation}
    \Delta_m^{(t+1)}={\mathbf{g}_m^{ec^{(t)}}}-{\mathbf{g}_m^{sp^{(t)}}}. \label{Delta}
\end{equation}
After sparsification, only the gradient components with the highest impact on model update are retained. Then, ${\mathbf{g}_m^{sp^{(t)}}}$ is compressed and transmitted over the wireless MAC. With the sparsity, compressed sensing is applied at the PS for model aggregation, as described below.\\
\indent At each round $t$, a pseudo-random linear compression $\mathbf{A}^{(t)} \in \mathbb{R}^{s \times N}$ is assigned and shared between the PS and all the edge devices, where each device computes 
\begin{equation}
    {\mathbf{g}_m^{cp^{(t)}}}=\mathbf{A}^{(t)}{\mathbf{g}_m^{sp^{(t)}}} \label{compression}.
\end{equation}
In contrast to the work in \cite{amiri2020machine} and \cite{amiri2020federated} where $\mathbf{A}^{(t)}$ is IID Gaussian, we adopt a partial discrete cosine transform (DCT) matrix\footnote{$\mathbf{A} \in \mathbb{R}^{s\times N}$ is said to be a partial DCT matrix iff the rows of $\mathbf{A}$ are selected from the $N$-by-$N$ DCT matrix.} $\mathbf{A}^{(t)}$ instead to reduce the complexity of transmitter-side compression and receiver-side reconstruction. 
 Considering limited wireless bandwidth, the compression in (11) is supposed to largely reduce the length of the gradient vector, i.e., $s \ll d$.\\
\indent The transmit signal of device $m$ is
\begin{equation}
    \tilde{\mathbf{g}}_m^{(t)}=\Phi\left(\mathbf{g}_m^{(t)}\right) \triangleq \sqrt{\alpha_m^{(t)}} \frac{MK_m}{K}{\mathbf{g}_m^{cp^{(t)}}} \label{power control},
\end{equation}
where the power control coefficient $\alpha_m^{(t)}$ is chosen to satisfy the constraint \eqref{power constraint}. 
This scaling value $\sqrt{\alpha_m^{(t)}} \frac{MK_m}{K}$ in \eqref{power control} is attached to the compressed local gradient $\tilde{\mathbf{g}}_m^{(t)}$ and sent to the PS. Compared with the high-dimensional $\tilde{\mathbf{g}}_m^{(t)}$, the overhead of transmitting a scalar is negligible, so we assume that $\sqrt{\alpha_m^{(t)}} \frac{MK_m}{K}$ is transmitted in a noise-free manner. 
\\ \indent From \eqref{compression} and \eqref{power control}, the scaled received signal over the MAC in \eqref{MAC} is given by
\begin{align}
\mathbf{y}^{(t)} =
\mathbf{A}^{(t)} {\mathbf{x}}^{(t)} +  \mathbf{n}^{(t)} ,
\label{model2}
\end{align}
where $\mathbf{y}^{(t)} \triangleq \frac{K \tilde{\mathbf{y}}^{(t)}}{M \sum_{m=1}^M \sqrt{\alpha_m^{(t)}} K_m}$, $\mathbf{x}^{(t)} \triangleq \frac{\sum_{m=1}^M \sqrt{\alpha_m^{(t)}} K_m {\mathbf{g}_m^{sp^{(t)}}}}{\sum_{m=1}^M \sqrt{\alpha_m^{(t)}} K_m}$ and $\mathbf{n}^{(t)} \triangleq \frac{K \tilde{\mathbf{n}}^{(t)}}{M \sum_{m=1}^M \sqrt{\alpha_m^{(t)}} K_m} \sim \mathcal{N}\left(0, \sigma^2 \right)$ with $\sigma \triangleq \frac{K \sigma_e} {M \sum_{m=1}^M \sqrt{\alpha_m^{(t)}K_m}} $. We assume that $\{\alpha_m^{(t)}\}_{m=1}^{M}$ are appropriately adjusted to satisfy $\alpha_1^{(t)}=...=\alpha_M^{(t)}$ for any $t$ and the power constraint \eqref{power constraint}. Then, 
\begin{equation}
    \mathbf{x}^{(t)} = \frac{\sum_{m=1}^M  K_m {\mathbf{g}_m^{sp^{(t)}}}}{\sum_{m=1}^M  K_m} \label{14}
\end{equation}
is the (sparsified) gradient aggregation required in \eqref{Aggregation}.
Let $\hat{\mathbf{x}}^{(t)}$ be an estimation of ${\mathbf{x}^{(t)}}$ reconstructed by the PS. Then, the global model is updated by
\begin{equation}
    \mathbf{w}^{(t+1)} = \mathbf{w}^{(t)} -\eta^{(t)} \mathbf{\hat{x}}^{(t)}. \label{31}
\end{equation}
\indent For the above reconstruction problem, the compressed sensing approach based on the well-known AMP algorithm refered to as A-DSGD is proposed in \cite{amiri2020machine}, where the PS recovers $\mathbf{x}^{(t)}$ only exploiting the sparsity of \eqref{14}. However, from a broader Bayesian perspective, this compressed-sensing based approach can be extended to take into consideration a finer structure of $\{\mathbf{x}^{(t)}\}_{t=1}^{T}$ over time. The inherent temporal correlation of gradient aggregation in both support and amplitude is validated and characterized formally in the next section. Thus, the PS, aware of the temporal structure of $\{\mathbf{x}^{(t)}\}_{t=1}^{T}$, estimates $\mathbf{x}^{(t)}$ based not only on the current observation but also on the other historical ones, i.e., the goal is to infer $\mathbf{x}^{(t)}$ in an online manner given all the past observations $\{{\mathbf{y}}^{(j)}\}_{j=1}^{(t)}$, which is referred to as TSA-GA. Compared with A-DSGD, the proposed TSA-GA scheme is able to achieve substantial enhancement of the recovery accuracy and acceleration of the over-the-air learning, as detailed in what follows.
\section{Probability Model of Gradient Aggregation}
In this section, we characterize the temporal structure of the gradient aggregation series.
\subsection{Temporal Structure of Gradient Aggregation}
\begin{figure}[!t]
\centering
\includegraphics [scale=0.45] {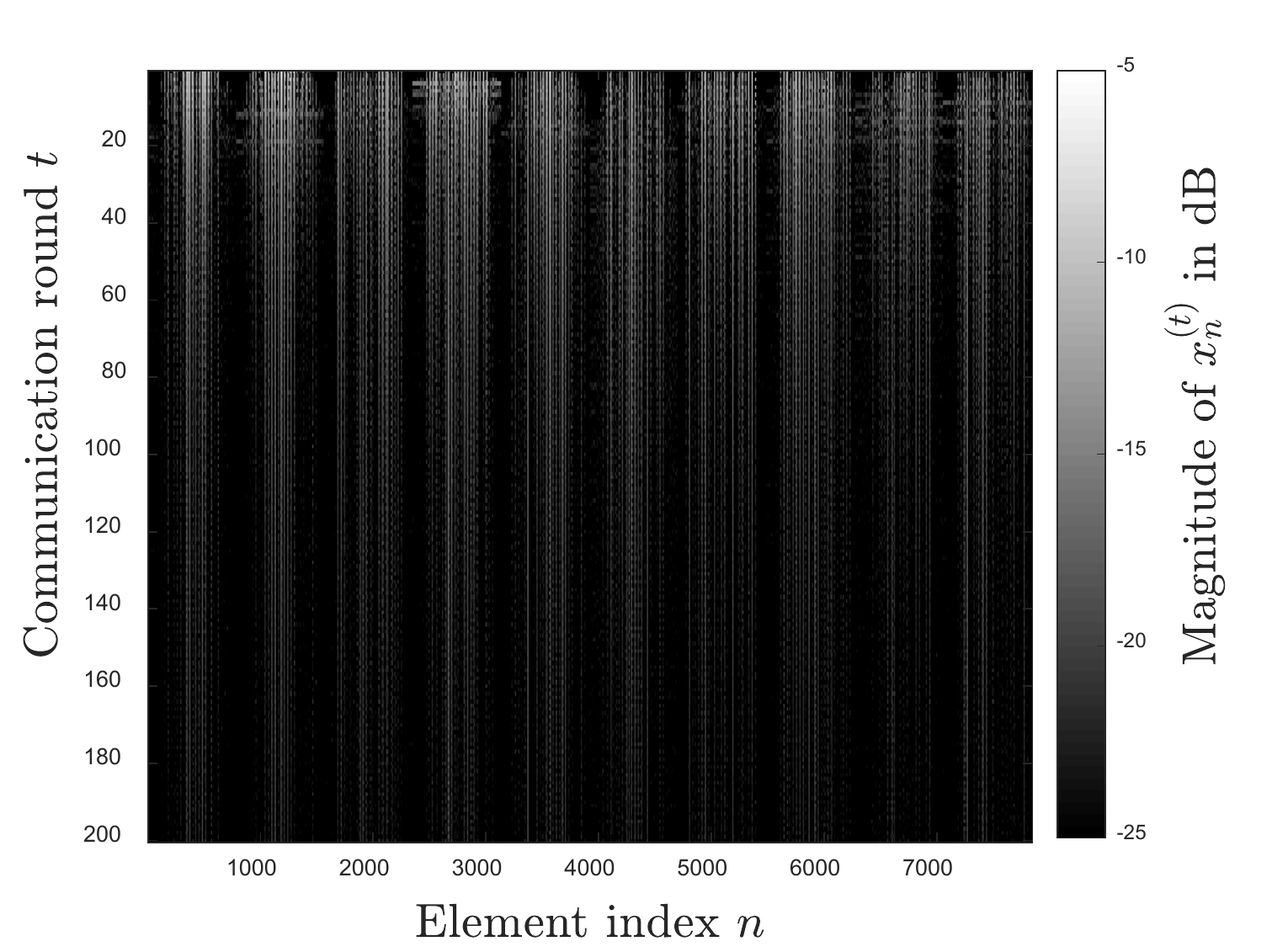}
\caption{Magnitude (in dB) of $x_n^{(t)}$ for the federated classification task on the MNSIT dataset where a single layer neural network is trained among $M=25$ devices using local gradient descent. $K_m=1000$, $\eta^{(t)}=0.01$, $E=1$, $k$ = 0.2.}
\end{figure}
To investigate the structure of $\{\mathbf{x}^{(t)}\}_{t=1}^{T}$, in Fig 2, we plot the magnitude of the elements of $\mathbf{x}^{(t)}$ during a FEEL training process for the digit classification task on the MNSIT dataset, where a global neural network with a single fully-connected layer is trained between $M=25$ devices and a PS. From Fig. 2, we have the following basic observation.\vspace{2ex}\\
\underline{\emph{Observation 1}} At each round $t$, many elements of $\mathbf{x}^{(t)}$ are close to zero, leading to sparse $\mathbf{x}^{(t)}$.\vspace{2ex}\\
The sparsity of $\mathbf{x}^{(t)}$ has already been investigated and exploited by many previous studies, such as for gradient compression and reconstruction improvement in \cite{jeon2020compressive,amiri2020machine}. To incorporate this sparsity of $\mathbf{x}^{(t)}$, one can adopt an IID Bernoulli-Gaussian prior on the elements of $\mathbf{x}^{(t)}$, i.e.,
\begin{equation}
    p\left(x_n^{(t)}\right)=\left(1-\lambda\right)\delta\left(x_n^{(t)}\right)+\lambda\mathcal{N}\left(x_n^{(t)};0,\gamma\right), \label{overall}
\end{equation}
for $t=1,2,...,T$, where $\lambda$ and $\gamma$ are the prior probability and the variance of the non-zero elements in $\mathbf{x}^{(t)}$, respectively. 
 Intuitively, the sparsity of $\mathbf{x}^{(t)}$ implies a large number of weak links between the neutrons at round $t$, with the sporadic large elements $x_n^{(t)}$ indicating the strong links that have significant impacts on the learning model. However, besides the sparsity of $\mathbf{x}^{(t)}$ at each round $t$, the strength of each link also exhibits strong correlation in time during the training process, as observed in Fig. 2. This phenomenon has also been observed in the field of model compression, e.g., weight pruning \cite{han2015learning}. To characterize a more complicated structure of $\{\mathbf{x}^{(t)}\}_{t=1}^{T}$, we decouple the $n$-th element of $\mathbf{x}^{(t)}$, $x_n^{(t)}$, by its support and amplitude\footnote{The amplitude contains the sign of $ x_n^{(t)}$.} as
\begin{equation}
    x_n^{(t)}=s_n^{(t)}r_n^{(t)} \label{decomp},
\end{equation}
where $s_n^{(t)} \in \{0,1\}$ indicates $x_n^{(t)}=0$ if $s_n^{(t)}=0$ and $x_n^{(t)} \neq0$ otherwise. From Fig. 2 , we have more observations on $\{\mathbf{x}^{(t)}\}_{t=1}^{T}$ as follows. \vspace{2ex}\\
\underline{\emph{Observation 2}}. The support $\mathbf{s}^{(t)}\triangleq [s_1^{(t)}, ...,s_N^{(t)}]$ varies slowly over  $t$.\\
\underline{\emph{Observation 3}}. The amplitude $\mathbf{r}^{(t)}\triangleq [r_1^{(t)}, ...,r_N^{(t)}]$ is highly correlated over $t$. \vspace{2ex}\\
\indent We remark the underlying insight of these observations. As stated before, the indices of these small model parameters correspond to weak links between the network neutrons. Furthermore, Observation 2 suggests that the locations of these weak links are also relatively static over time. Besides, as the loss function in a learning task is typically optimized with a relatively small learning rate, Observation 3 shows that the significant elements of the gradient evolve step by step up to a small variation. These phenomena give rise to a finer prior model for each $x_n^{(t)}$ by incorporating two independent Markov chains for characterizing the temporal evolution of its support $s_n^{(t)}$ and amplitude $r_n^{(t)}$ respectively.
\\ \indent To incorporate the time correlation of $s_n^{(t)}$ in Observation 2, for each $n$, $\{s_n^{(t)}\}_{t=1}^T$  is modeled as an independent Markov chain as
\begin{equation}
    p(s_n^{(t)}|s_n^{(t-1)})=
\begin{cases}
(1-p_{10})^{1-s_n^{(t)}}p_{10}^{s_n^{(t)}}, & s_n^{(t-1)}=0,\\
p_{01}^{1-s_n^{(t)}}(1-p_{01})^{s_n^{(t)}}, & s_n^{(t-1)}=1 \label{support},
\end{cases}
\end{equation}
with transition probabilities $ p_{01} = \mathbb{P}\{s_n^{(t)}=0|s_n^{(t-1)}=1\}$ and $p_{10}= \mathbb{P}\{s_n^{(t)}=1|s_n^{(t-1)}=0\}$ for $t=2,...,T$. Typically, $p_{01}$ and $p_{10}$ tend to be very small since the essential links in the neural network remain almost the same over time, leading to slow transition in the support of $\mathbf{x}^{(t)}$. 
\\ \indent Likewise, the time correlation of each $\{r_n^{(t)}\}_{t=1}^T$ for each $n$ can be modeled as an independent order-1 auto-regression progress as
\begin{equation}
    r_n^{(t)}=\left(1-\beta\right) r_n^{(t-1)}+\beta \omega_n^{(t)} \label{amplitude},
\end{equation}
for $t=2,...,T$. In the above, $\omega_n^{(t)} \sim \mathcal{N}\left(0,\xi\right)$ for all $n$ are IID Gaussian perturbations, and $\beta \in [0,1]$  controls the auto-correlation over time. All the Markov chains in \eqref{support} and \eqref{amplitude} for $n=1,2,...,N$ are assumed to be independent.\\
\indent We remark that similar Bayesian modelling techniques have been previously used in \cite{ziniel2012efficient,chen2017structured} for other signal estimation tasks. Yet,  
to the best of our knowledge, this work is the first attempt to exploit the temporal structure to assist the gradient aggregation in the FEEL scenario. 
\subsection{Problem Formulation}
Recall that the task of the PS in each round $t$ is to reconstruct $\mathbf{x}^{(t)}$ given $\{{\mathbf{y}}^{(j)}\}_{j=1}^{(t)}$. Based on the probability model above, we formulate this task as an online Bayesian inference problem. Specifically, from \eqref{overall}, \eqref{support}, \eqref{amplitude} and \eqref{model2}, the joint posterior probabilistic density function (PDF) of $\mathbf{x}^{(t)}$,$\{\mathbf{s}^{(j)}\}_{j=1}^t,\{\mathbf{r}^{(j)}\}_{j=1}^t$ given $\mathbf{y}^{(t)}$ can be expressed as
\begin{multline}
    p\left(\{{\mathbf{x}}^{(j)}\}_{j=1}^t,\{\mathbf{s}^{(j)}\}_{j=1}^t,\{\mathbf{r}^{(j)}\}_{j=1}^t|\{\mathbf{y}^{(j)}\}_{j=1}^t \right)\\ \propto 
    \prod_{j=1}^t p\left(\mathbf{y}^{(j)}|\mathbf{x}^{(j)}\right) \prod_{n=1}^{N} p\left(x_n^{(j)}|s_n^{(j)},r_n^{(j)}\right)
    p\left(s_n^{(j)}|s_n^{(j-1)}\right) p\left(r_n^{(j)}|r_n^{(j-1)}\right). \label{post}
\end{multline}
The classic MMSE estimator $\mathbb{E}\left[\mathbf{x}^{(t)}|\{{\mathbf{y}}^{(j)}\}_{j=1}^{(t)}\right]$ for $\mathbf{x}^{(t)}$
is computationally prohibitive. Note that \eqref{post} can be represented by a factor graph as in Fig. 3, where variable nodes appear as white circles and factor nodes appear as black boxes. There is an edge connection between a variable node and a factor node if the variable node appears in the corresponding factor. Based on this factor graph representation, we next present a message-passing based algorithm to approximate the MMSE estimator of $\mathbf{x}^{(t)}$ with low complexity and near-optimal performance.

\begin{figure}[!t]
\centering
\includegraphics [scale=0.40] {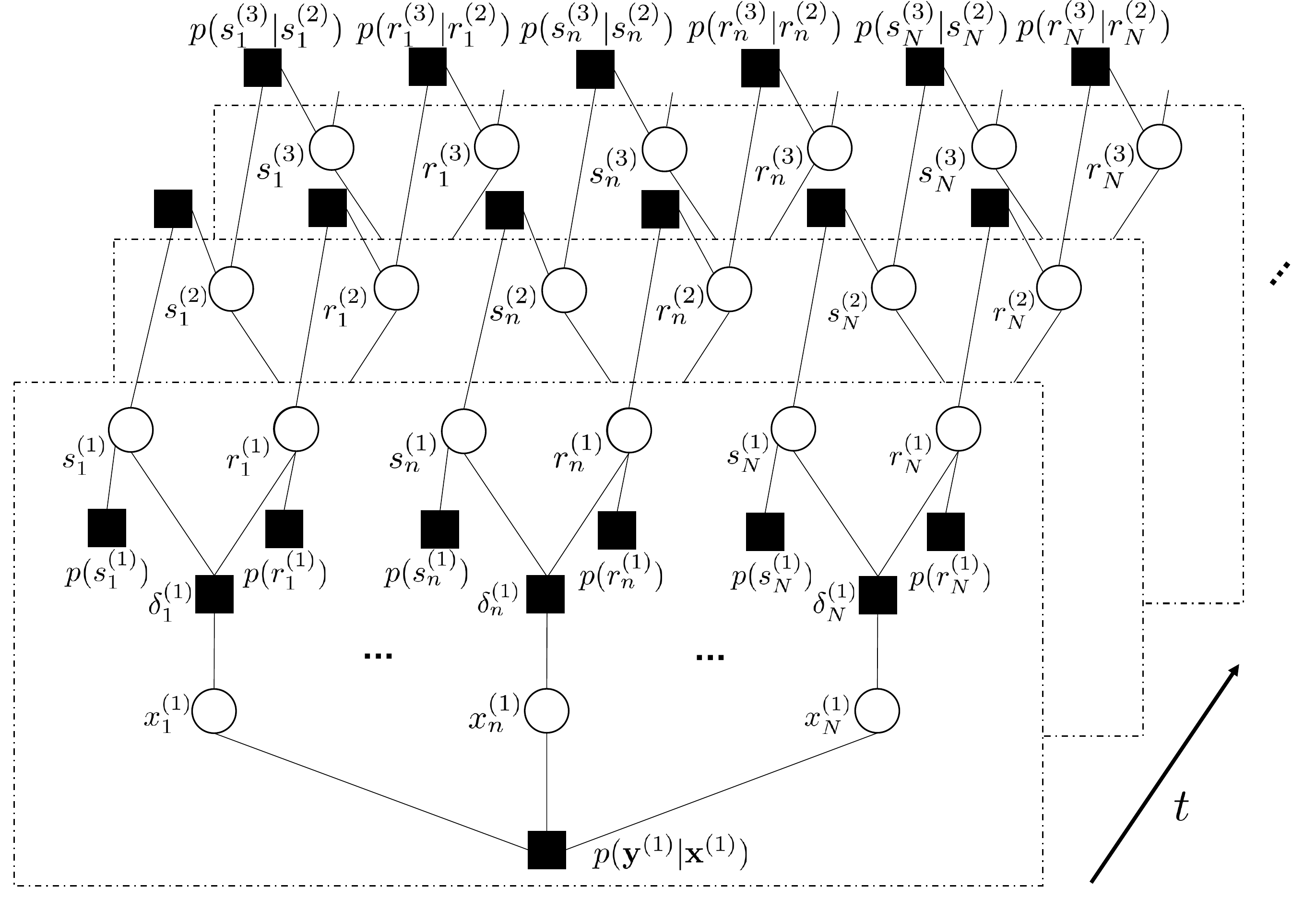}
\caption{Factor graph representation of the PDF \eqref{post}, where we adopt the abbreviation $\delta_n^{(t)}$ for $p(x_n^{(t)}|s_n^{(t)},r_n^{(t)})$.}
\label{factor graph}
\end{figure}

\section{Gradient Aggregation Recovery Algorithm}
In this section, we propose a message-passing based solution, referred to as TSA-GA, to the aforementioned online Bayesian inference problem.
 The proposed algorithm for recovering gradient aggregation is an online extension of the Turbo-CS algorithm in \cite{ma2014turbo} by incorparating the underlying Markovian prior of $\{\mathbf{x}^{(t)}\}_{t=1}^T$. 

\subsection{Algorithm Design}
We now derive the turbo message passing algorithm by following the sum-product rule over the graph in Fig. 3 \cite{pearl2014probabilistic,kschischang2001factor}.
\subsubsection{Messages between $x_n^{(t)}$ and $p(\mathbf{y}^{(t)}|\mathbf{x}^{(t)})$}
For notation brevity, we write node $p(\mathbf{y}^{(t)}|\mathbf{x}^{(t)})$ as $\mathbf{y}^{(t)}$ and $\delta_n^{(t)}$ as $p(x_n^{(t)}|s_n^{(t)},r_n^{(t)})$ in the following derivation. At each round $t$, based on observation $\mathbf{y}^{(t)}$ and the messages from $\delta_n^{(t)}$ to $x_{n}^{(t)}$
(i.e., $\{\nu_{\delta_n^{(t)} \rightarrow x_{n}^{(t)}} (x_{n}^{(t)})\}_{n=1}^N$), an estimation of $\mathbf{x}^{(t)}$ is obtained via turbo message passing between the node $p\left(\mathbf{y}^{(t)}|\mathbf{x}^{(t)}\right)$ and $\mathbf{x}^{(t)}$ \cite{ma2014turbo}. Denote by $\nu_{\mathbf{y}^{(t)} \leftarrow x_{n}^{(t)}}(x_{n}^{(t)})$ the message from $x_{n}^{(t)}$ to $\mathbf{y}^{(t)}$. By the sum-product rule, the message from $\mathbf{y}^{(t)}$ to $x_{n}^{(t)}$ is given by
\begin{equation}
    \nu_{\mathbf{y}^{(t)} \rightarrow x_{n}^{(t)}} (x_{n}^{(t)}) \propto
    \frac{ \int_{ \{x_{i}^{(t)}\}_{i \neq n} } \mathcal{N}(\mathbf{y}^{(t)}; \mathbf{A}^{(t)}\mathbf{x}^{(t)}, \sigma^2 \mathbf{I}) 
    \prod_{i=1}^{N} \nu_{\mathbf{y}^{(t)} \leftarrow x_{i}^{(t)}}(x_{i}^{(t)}) 
     } 
    { \nu_{\mathbf{y}^{(t)} \leftarrow x_{n}^{(t)}}(x_{n}^{(t)}) }. \label{rule1}
\end{equation}
We assume that $\nu_{\mathbf{y}^{(t)} \leftarrow x_{n}^{(t)}}(x_{n}^{(t)})$ is Gaussian with mean $\mu_{\mathbf{y}^{(t)} \leftarrow x_{n}^{(t)}}$ and variance $v_{\mathbf{y}^{(t)} \leftarrow \mathbf{x}^{(t)}}$. The numerator in \eqref{rule1} then reduces to $\mathcal{N}(x_{n}^{(t)};\mu_{\mathbf{y}^{(t)},n},v_{\mathbf{y}^{(t)}})$, where 
\begin{align}
        \boldsymbol{\mu}_{\mathbf{y}^{(t)}} & = \boldsymbol{\mu}_{\mathbf{y}^{(t)} \leftarrow x_{n}^{(t)}} + 
    \frac{ 
    v_{\mathbf{y}^{(t)} \leftarrow \mathbf{x}^{(t)}}}
    { v_{\mathbf{y}^{(t)} \leftarrow \mathbf{x}^{(t)}} + \sigma^2 }
   (\mathbf{A}^{(t)})^{T} (\mathbf{y}^{(t)} - \mathbf{A}^{(t)} \boldsymbol{\mu}_{\mathbf{y}^{(t)} \leftarrow x_{n}^{(t)}}) \label{mu_1},\\
       v_{\mathbf{y}^{(t)}} &= v_{\mathbf{y}^{(t)} \leftarrow \mathbf{x}^{(t)}} - 
    \frac{N}{s} \frac{v_{\mathbf{y}^{(t)} \leftarrow \mathbf{x}^{(t)}}^2} 
    { v_{\mathbf{y}^{(t)} \leftarrow \mathbf{x}^{(t)}} + \sigma^2 }, \label{v_1}
\end{align}
and $\mu_{\mathbf{y}^{(t)},n}$ is the $n$-th element of $ \boldsymbol{\mu}_{\mathbf{y}^{(t)}}$.
Here we have used the partial orthogonal property $\mathbf{A}^{(t)} \left(\mathbf{A}^{(t)}\right)^T=\mathbf{I}$ to simplify the expression. 

Plugging \eqref{mu_1} and \eqref{v_1} into \eqref{rule1}, we obtain 
\begin{equation}
    \nu_{\mathbf{y}^{(t)} \rightarrow x_{n}^{(t)}} (x_{n}^{(t)}) \propto
    \mathcal{N}\left(x_{n}^{(t)}; \mu_{\mathbf{y}^{(t)} \rightarrow x_n^{(t)}}, 
    v_{\mathbf{y}^{(t)}\rightarrow \mathbf{x}^{(t)}}\right) \label{rule11},
\end{equation}
where 
\begin{align}
        \mu_{\mathbf{y}^{(t)} \rightarrow x_n^{(t)}} &= v_{\mathbf{y}^{(t)}\rightarrow \mathbf{x}^{(t)}} \left(\frac{\mu_{\mathbf{y}^{(t)},n}}{v_{\mathbf{y}^{(t)}}} - \frac{\mu_{\mathbf{y}^{(t)} \leftarrow x_{n}^{(t)}}}{v_{\mathbf{y}^{(t)} \leftarrow \mathbf{x}^{(t)}}} \right), \\
    v_{\mathbf{y}^{(t)}\rightarrow \mathbf{x}^{(t)}} &= \left(\frac{1}{v_{\mathbf{y}^{(t)}}} - \frac{1}{v_{\mathbf{y}^{(t)} \leftarrow \mathbf{x}^{(t)}}} \right)^{-1}. \label{v_11}
\end{align}
Combining \eqref{rule1} and \eqref{rule11}, the posterior message of $x_{n}$ is approximated by
\begin{equation}
\nu_{\mathbf{y}^{(t)} \rightarrow x_{n}^{(t)}}(x_{n}^{(t)}) \nu_{\delta_n^{(t)} \rightarrow x_{n}^{(t)}} (x_{n}^{(t)}). \label{post msg}
\end{equation}
The posterior mean $\mu_n$ and variance $v_n$ of $x_{n}$ at iteration $i$ are given respectively by
\begin{align}
        \mu_n^{(t)} &= \mathbb{E}\left[x_{n}^{(t)}
        \right],\\
    v^{(t)} &=  \frac{1}{N} \sum_{n=1}^N  \text{Var}\left[x_{n}^{(t)}
    \right],
\end{align}
where the expectation and variance of $x_n^{(t)}$ are with respect to the posterior message in \eqref{post msg}. From \cite{ma2014turbo,ma2015performance}, 
the mean $\mu_{\mathbf{y}^{(t)} \leftarrow x_{n}^{(t)}}$ and variance $v_{\mathbf{y}^{(t)} \leftarrow \mathbf{x}^{(t)}}$ of $\nu_{\mathbf{y}^{(t)} \leftarrow x_{n}^{(t)}}(x_{n}^{(t)})$ are given by 
\begin{align}
        \mu_{\mathbf{y}^{(t)} \leftarrow x_{n}^{(t)}} &=  v_{\mathbf{y}^{(t)} \leftarrow \mathbf{x}^{(t)}}
    \left( \frac{\mu_n^{(t)}}{v^{(t)}}
    -  \frac{\mu_{\mathbf{y}^{(t)} \rightarrow x_n^{(t)}}}{v_{\mathbf{y}^{(t)}\rightarrow \mathbf{x}^{(t)}}}
    \right), \\
    v_{\mathbf{y}^{(t)} \leftarrow \mathbf{x}^{(t)}} &= 
     \left( \frac{1}{v^{(t)}}
    - \frac{1}{v_{\mathbf{y}^{(t)}\rightarrow \mathbf{x}^{(t)}}} \right).
\end{align}
When a predetermined termination condition is met, $\mathbf{\hat{x}}^{(t)}= [\mu_1^{(t)},\mu_2^{(t)},...,\mu_N^{(t)}]^T$ is returned as the gradient aggregation estimator at the $t$-th communication round. Then, the global FEEL model is updated via \eqref{31}
and broadcast to all the participating edge devices afterwards. 
 \subsubsection{Messages from $\delta_n^{(t)}$ to $s_n^{(t)}$ and $r_n^{(t)}$}
Messages of the $t$-th round are propagated forward to provide the prior knowledge for the $(t+1)$-th round. To this end, we calculate the message along the Markov chains based on the prior models \eqref{support} and \eqref{amplitude}. 
From the sum-product rule, the message from $\delta_n^{(t)}$ to $s_n^{(t)}$ is given by
\begin{align}
   \nu_{\delta_n^{(t)}\rightarrow s_n^{(t)}}\left(s_n^{(t)}\right) 
    \propto {}&
    \int_{x_n^{(t)},r_n^{(t)}} \nu_{{r_n^{(t-1)}}\rightarrow{r_n^{(t)}}} (r_n^{(t)})
     \nu_{\mathbf{y}^{(t)} \rightarrow x_{n}^{(t)}} (x_{n}^{(t)})
    \delta \left(x_n^{(t)}-s_n^{(t)}r_n^{(t)}\right) \notag \\
    \propto {}&     
    \lambda_{\delta_n^{(t)}\rightarrow s_{n}^{(t)}}
    \delta \left(s_n^{(t)} -1\right) +
    \left(1-\lambda_{\delta_n^{(t)}\rightarrow s_{n}^{(t)}}\right)
    \delta\left(s_n^{(t)}\right) \label{delta-s},
\end{align}
with
\begin{equation}
        \lambda_{\delta_n^{(t)}\rightarrow s_{n}^{(t)}} = \left( 1+ 
        \frac{\mathcal{N}\left(0;\mu_{\mathbf{y}^{(t)} \rightarrow x_n^{(t)}}, v_{\mathbf{y}^{(t)}\rightarrow \mathbf{x}^{(t)}}\right)}{
        \mathcal{N}\left(0;\mu_{r_n^{(t-1)} \rightarrow r_n^{(t)}}-\mu_{\mathbf{y}^{(t)} \rightarrow x_n^{(t)}}, v_{r_n^{(t-1)} \rightarrow r_n^{(t)}}+v_{\mathbf{y}^{(t)}\rightarrow \mathbf{x}^{(t)}}\right)}
        \right)^{-1}, 
\end{equation}
where $\nu_{{r_n^{(t-1)}}\rightarrow{r_n^{(t)}}} (r_n^{(t)})=\mathcal{N}(r_n^{(t)};\mu_{r_n^{(t-1)} \rightarrow r_n^{(t)}}, v_{r_n^{(t-1)} \rightarrow r_n^{(t)}})$ is the message passed from $r_n^{(t-1)}$ to $r_n^{(t)}$ via the node $p(r_n^{(t)}|r_n^{(t-1)})$. The message from $\delta_n^{(t)}$ to $r_n^{(t)}$, however, cannot be calculated straightforwardly. 
 This is due to the fact that the inference of $r_n^{(t)}$ is infeasible given  $x_n^{(t)}$ when $s_n^{(t)}=0$, by recalling the model in \eqref{decomp}. To circumvent this difficulty, we follow the modification in \cite{ziniel2012efficient} and regard the model in \eqref{decomp} as a limiting result of $s_n^{(t)} \in \{\epsilon,1\}$ as $\epsilon \rightarrow 0^+$. In this view, the message from $\delta_n^{(t)}$ to $r_n^{(t)}$ for any fixed $\epsilon > 0$ is given by 
\begin{align}
    \nu_{\delta_n^{(t)}\rightarrow r_n^{(t)}}(r_n^{(t)})
    \propto {}&
    \int_{x_n^{(t)},s_n^{(t)}}  \nu_{{s_n^{(t-1)}}\rightarrow{s_n^{(t)}}}({s_n^{(t)}})
\nu_{\mathbf{y}^{(t)} \rightarrow x_{n}^{t}} (x_{n}^{(t)})
    \delta(x_n^{(t)}-s_n^{(t)}r_n^{(t)}) \notag \\
    \propto {} &
    (1-\Omega( \lambda_{s_n^{(t-1)}\rightarrow s_{n}^{(t)}})) \mathcal{N}\left(r_n^{(t)};\frac{1}{\epsilon} \mu_{\mathbf{y}^{(t)} \rightarrow x_n^{(t)}},\frac{1}{\epsilon^2} v_{\mathbf{y}^{(t)}\rightarrow \mathbf{x}^{(t)}}\right) \notag \\
    {} & +
    \Omega( \lambda_{s_n^{(t-1)}\rightarrow s_{n}^{(t)}}) \mathcal{N}\left(r_n^{(t)}; \mu_{\mathbf{y}^{(t)} \rightarrow x_n^{(t)}}, v_{\mathbf{y}^{(t)}\rightarrow \mathbf{x}^{(t)}}\right) \label{mixture},
\end{align}
with
\begin{equation}
    \Omega\left( \lambda_{s_n^{(t-1)}\rightarrow s_{n}^{(t)}}\right) = \frac{\epsilon  \lambda_{s_n^{(t-1)}\rightarrow s_{n}^{(t)}}}
                {\left(1- \lambda_{s_n^{(t-1)}\rightarrow s_{n}^{(t)}} \right) + \epsilon     \lambda_{s_n^{(t-1)}\rightarrow s_{n}^{(t)}}},
\end{equation}
where $  \nu_{{s_n^{(t-1)}}\rightarrow{s_n^{(t)}}}({s_n^{(t)}}) = \lambda_{s_n^{(t-1)}\rightarrow s_{n}^{(t)}} 
    \delta(s_n^{(t)} -1) +
    (1- \lambda_{s_n^{(t-1)}\rightarrow s_{n}^{(t)}})
    \delta(s_n^{(t)})$ is the message passed from $s_n^{(t-1)}$ to $s_n^{(t)}$ via the node $p(s_n^{(t)}|s_n^{(t-1)})$.
As $\epsilon \rightarrow 0^+$, the first component in the last step of \eqref{mixture} tends to be non-informative.
Therefore, to keep \eqref{mixture} meaningful, we set $\epsilon$ to a sufficiently small positive number, e.g. $10^{-7}$. 
Yet, it is still undesirable that \eqref{mixture} takes the form of a Gaussian mixture. Propagating this message exactly will cause an exponential computation burden because each round of forward propagation doubles the number of Gaussian components. Here, we Taylor-expand the logarithm of the message in \eqref{mixture} with respect to $r_n^{(t)}$ at $r_n^{(t)}=\mu_{\mathbf{y}^{(t)} \rightarrow x_n^{(t)}}$ until the second order, yielding a Gaussian approximation 
\begin{equation}
    \nu_{\delta_n^{(t)}\rightarrow r_n^{(t)}}(r_n^{(t)}) \approx \mathcal{N}(r_n^{(t)};
    {\bar{\mu}_n^{(t)}},\bar{v}_n^{(t)}),
\end{equation}
with
\begin{align}
    \bar{\mu}_n^{(t)}\triangleq \mu_{\mathbf{y}^{(t)} \rightarrow x_n^{(t)}}-\left(\frac{d^2 \ln \nu_{\delta_n^{(t)}\rightarrow r_n^{(t)}}}{dr_n^{(t)2}}\right)^{-1}\frac{d\ln \nu_{\delta_n^{(t)}\rightarrow r_n^{(t)}}}{dr_n^{(t)}} \!\text{\!\quad and\quad\!}\!
    \bar{v}_n^{(t)} \triangleq -\left(\frac{d^2 \ln \nu_{\delta_n^{(t)}\rightarrow r_n^{(t)}}}{dr_n^{(t)2}}\right)^{-1}.
\end{align}

\subsubsection{Messages from $s_n^{(t)}$ to $s_n^{(t+1)}$ and from $r_n^{(t)}$ to $r_n^{(t+1)}$}
The message passed from $s_n^{(t)}$ to $s_n^{(t+1)}$ via the node $p(s_n^{(t+1)}|s_n^{(t)})$ is given by
\begin{align}
    \nu_{{s_n^{(t)}}\rightarrow{s_n^{(t+1)}}}({s_n^{(t+1)}})
    \propto{}& 
    \int_{s_n^{(t)}}    
     p(s_n^{(t+1)}|s_n^{(t)})
    \left(\lambda_{\delta_n^{(t)}\rightarrow s_{n}^{(t)}} 
    \delta(s_n^{(t)} -1) +
    (1-\lambda_{\delta_n^{(t)}\rightarrow s_{n}^{(t)}})
    \delta(s_n^{(t)})\right)
      \notag \\
    {} & \times
     \left(\lambda_{s_n^{(t-1)}\rightarrow s_{n}^{(t)}} 
    \delta(s_n^{(t)} -1) +
    (1- \lambda_{s_n^{(t-1)}\rightarrow s_{n}^{(t)}})
    \delta(s_n^{(t)})\right) \notag \\
     \propto{}&
     \lambda_{s_n^{(t)}\rightarrow s_{n}^{(t+1)}} 
    \delta(s_n^{(t+1)} -1) +
    (1- \lambda_{s_n^{(t)}\rightarrow s_{n}^{(t+1)}})
    \delta(s_n^{(t+1)}) \label{across1},
\end{align}
where
\begin{equation}
        \lambda_{s_n^{(t)}\rightarrow s_{n}^{(t+1)}} = \frac{p_{10} (1-\lambda_{s_n^{(t-1)}\rightarrow s_{n}^{(t)}})(1-\lambda_{\delta_n^{(t)}\rightarrow s_{n}^{(t)}}
    + (1-p_{01})\lambda_{s_n^{(t-1)}\rightarrow s_{n}^{(t)}} \lambda_{\delta_n^{(t)}\rightarrow s_{n}^{(t)}}}
    {(1 - \lambda_{s_n^{(t-1)}\rightarrow s_{n}^{(t)}}) (1-\lambda_{\delta_n^{(t)}\rightarrow s_{n}^{(t)}}) + \lambda_{s_n^{(t-1)}\rightarrow s_{n}^{(t)}} \lambda_{\delta_n^{(t)}\rightarrow s_{n}^{(t)}}}.
\end{equation}
Similarly, the message from $r_n^{(t)}$ to $r_n^{(t+1)}$ via the node $p(r_n^{(t+1)}|r_n^{(t)})$ is given by
\begin{align}
    \nu_{{r_n^{(t)}}\rightarrow{r_n^{(t+1)}}} (r_n^{(t+1)})
    \propto{}& 
    \int_{r_n^{(t)}}    
     p(r_n^{(t+1)}|r_n^{(t)})
    \mathcal{N}(r_n^{(t)}; \mu_{\delta_{n}^{(t)} \rightarrow r_n^{(t)}},v_{\delta_{n}^{(t)} \rightarrow r_n^{(t)}}) 
        \mathcal{N}(r_n^{(t)};\mu_{r_n^{(t-1)} \rightarrow r_n^{(t)}}
    ,v_{r_n^{(t-1)} \rightarrow r_n^{(t)}})
      \notag \\
     \propto{}&
    \mathcal{N}(r_n^{(t+1)};\mu_{r_n^{(t)} \rightarrow r_n^{(t+1)}},v_{r_n^{(t)} \rightarrow r_n^{(t+1)}} ), \label{across2}
\end{align}
with
\begin{gather}
    \mu_{r_n^{(t)} \rightarrow r_n^{(t+1)}} = (1-\beta) 
                \frac{v_{r_n^{(t-1)} \rightarrow r_n^{(t)}}  v_{\delta_{n}^{(t)} \rightarrow r_n^{(t)}}}
                {v_{r_n^{(t-1)} \rightarrow r_n^{(t)}}  + v_{\delta_{n}^{(t)} \rightarrow r_n^{(t)}}} 
                \left(\frac{\mu_{\delta_{n}^{(t)} \rightarrow r_n^{(t)}}}  {{v_{\delta_{n}^{(t)} \rightarrow r_n^{(t)}}}} +
                \frac{\mu_{r_n^{(t-1)} \rightarrow r_n^{(t)}}} {v_{r_n^{(t-1)} \rightarrow r_n^{(t)}}}\right)\\
    v_{r_n^{(t)} \rightarrow r_n^{(t+1)}} = (1-\beta)^2
              \frac{v_{r_n^{(t-1)} \rightarrow r_n^{(t)}}  v_{\delta_{n}^{(t)} \rightarrow r_n^{(t)}}}
                {v_{r_n^{(t-1)} \rightarrow r_n^{(t)}}  + v_{\delta_{n}^{(t)} \rightarrow r_n^{(t)}}}  +\beta^2 \xi.
\end{gather}
 From \eqref{across1} and \eqref{across2}, the message from $\delta_n^{(t+1)}$ to $x_{n}^{(t+1)}$ is given by 
\begin{multline}
    \nu_{\delta_n^{(t+1)} \rightarrow x_{n}^{(t+1)}} (x_{n}^{(t+1)}) = \\
    \begin{cases}
     \lambda_{s_n^{(t)}\rightarrow s_{n}^{(t+1)}}\mathcal{N}(x_n^{(t)};\mu_{r_n^{(t)} \rightarrow r_n^{(t+1)}},v_{r_n^{(t)} \rightarrow r_n^{(t+1)}}) + (1-\lambda_{s_n^{(t)}\rightarrow s_{n}^{(t+1)}}) \delta(x_n^{(t)}), & \text{if } t\geq 1,\\
    \lambda \mathcal{N}(x_n^{(1)};0,\gamma) +
    (1-\lambda) \delta(x_n^{(1)}), & \text{if } t=0.\\ 
    \end{cases}
\end{multline}
\subsection{Overall Algorithm}
The overall algorithm is summarized in Algorithm 1.
The computational complexity of Algorithm 1 is dominated by the multiplications associated with $\mathbf{A}$ and $\mathbf{A}^T$ in \eqref{mu_1}. Noting that $\mathbf{A}$ is a partial DCT matrix, we realize these multiplications by the fast DCT and inverse DCT algorithms with $\mathcal{O}(N\text{log}N)$ scalar manipulations. As a result, the total complexity of the proposed algorithm is bounded by $\mathcal{O}(I_{max}N\text{log}N)$ for each communication round, where $I_{max}$ is the pre-defined maximum number of turbo iterations. Compared with compression with IID Gaussian matrix in A-DSGD, the computational complexity of the proposed algorithm is more hardware-friendly thanks to the fast DCT algorithm. Moreover, the proposed algorithm exploits the intrinsic temporal stricture of the gradient signals and enables a more accurate and reliable recovery, thereby accelerating the convergence of the learning process. 

\begin{algorithm}[!t] 
\caption{TSA-GA Algorithm for Over-the-Air FEEL} 
\label{alg1} 
\begin{algorithmic}[1] 
\REQUIRE ~~\\ 
$\lambda, p_{01}, p_{10}, \gamma, \beta, \xi, \epsilon, I_{max}$\\
\FOR{$t=1,2,...,T$}
\STATE{PS receives $\mathbf{y}^{(t)}$}\\
\STATE Calculate the message from $\delta_n^{(t)}$ to $x_n^{(t)}$ via (43)\\
\STATE Initialize $v_{\mathbf{y}^{(t)} \leftarrow \mathbf{x}^{(t)}}= \gamma$ for $t=1$ and $v_{r_n^{(t-1)} \rightarrow r_n^{(t)}}$ otherwise, $\mu_{\mathbf{y}^{(t)} \leftarrow x_{n}^{(t)}}=0$ 
\WHILE{some terminal criterion is not met}
\STATE Calculate the message from $\mathbf{y}^{(t)}$ to $x_{n}^{(t)}$ via (22)(23)(25)(26) \% linear estimator
\STATE Calculate the message from $x_{n}^{(t)}$ to $\mathbf{y}^{(t)}$ via (28)(29)(30)(31) \% MMSE denoiser
\ENDWHILE
\STATE Global model $\boldsymbol{\theta}^{(t)}$ update via (15) and broadcast
\STATE Calculate the message from $\delta_n^{(t)}$ to $s_n^{(t)}$ via (33)
\STATE Calculate the message from $\delta_n^{(t)}$ to $r_n^{(t)}$ via (35)(37)
\STATE Calculate the message from $s_n^{(t)}$ to $s_n^{(t+1)}$ via (39)
\STATE Calculate the message from $r_n^{(t)}$ to $r_n^{(t+1)}$ via (41)(42)
\STATE Decide chain parameters for the next round \% by invoking Algorithm 2
\ENDFOR
\end{algorithmic}
\end{algorithm}

\section{Convergence Analysis}
In this section, we present the theoretical analysis of the TSA-GA algorithm. The main contribution is the development of the state evolution (SE) analysis for signals with underlying Markovian temporal structure. 

\subsection{State Evolution}
Recall the signal model of $\{\mathbf{x}^{(t)}\}_{t=1}^T$ to be estimated with PDF given by 
\begin{equation}
    p(\{\mathbf{x}^{(t)}\}_{t=1}^T)=\prod_{t=1}^{T} p(\mathbf{x}^{(t)}|\mathbf{x}^{(t-1)})=\prod_{n=1}^N \prod_{t=1}^{T}p(x_n^{(t)}|x_n^{(t-1)}) \label{model_SE},
\end{equation}
Note that each component of $\mathbf{x}^{(t)}$ evolves independently and identically in a Markovian manner. Thus, we simply refer to $x_n^{(t)}$ by omitting subscript as $x^{(t)}$ without causing confusion.\\ 
\indent In each communication round, the proposed online TSA-GA algorithm can be treated as the iteration between two modules, namely, the linear estimation module and the MMSE denoiser \cite{ma2014turbo}. The linear estimation module handles the linear constraint \eqref{model2}, and consists of the update equations (22) (23) (25) (26); see Line 6 of Algorithm 1. The MMSE denoiser suppresses the estimation error by exploiting the prior of $\mathbf{x}^{(t)}$, and consists of the update equations (28) (29) (30) (31); see Line 7 of Algorithm 1. The iteration between these two modules continues until convergence, as illustrated in Fig. 4(a). Based on this block diagram representation, for each round $t$, we define the states of the linear estimation module and the MMSE denoiser at iteration $i$ respectively as 
\begin{gather}
    \tau^{(t)}_i \triangleq \frac{1}{N}\Vert \boldsymbol{\mu}^{(t)}_{\mathbf{y} \rightarrow \mathbf{x}}-\mathbf{x}^{(t)} \Vert_2^2 \label{def_tau}, \\
    v^{(t)}_i \triangleq \frac{1}{N}\Vert \boldsymbol{\mu}^{(t)}_{\mathbf{y} \leftarrow \mathbf{x}}-\mathbf{x}^{(t)} \Vert_2^2 \label{def_v},
\end{gather}
where $\boldsymbol{\mu}^{(t)}_{\mathbf{y} \rightarrow \mathbf{x}} \triangleq [
\mu_{\mathbf{y}^{(t)} \rightarrow x_1^{(t)}},...,\mu_{\mathbf{y}^{(t)} \rightarrow x_N^{(t)}}]^{T}$ is the output of the linear estimation module and $\boldsymbol{\mu}^{(t)}_{\mathbf{y} \leftarrow \mathbf{x}} \triangleq [
\mu_{\mathbf{y}^{(t)} \leftarrow x_1^{(t)}},...,\mu_{\mathbf{y}^{(t)} \leftarrow x_N^{(t)}}]^{T}$ is the output of the MMSE denoiser at each round $t$. It is known from \cite{ma2015performance}
that, in the large system limit\footnote{Precisely, the large system limit means $N \rightarrow \infty$, $M \rightarrow \infty$ while $M/N$ is kept to constant.}, the output of the linear estimation module can be modelled as scalar observations of
\begin{equation}
    z_i^{(t)} = x^{(t)} + \sqrt{\tau_i^{(t)}} w_i^{(t)} \label{model3},
\end{equation}
where the subscript $i$ stands for the $i$-th turbo iteration and $w_i^{(t)} \sim \mathcal{N}(0,1)$ is the observation noise independent of the signal series $\{\mathbf{x}^{(t)}\}_{t=1}^T$. 
Thus, the behaviour of the turbo iteration is reduced to a scalar recursion between the two state variables:  
\begin{gather}
    \tau^{(t)}_i = f(v^{(t)}_i) \triangleq \left(\frac{N}{s}-1 \right) v^{(t)}_i + \sigma^2  \label{f},\\
    v^{(t)}_{i+1} = g_t(\tau^{(t)}_i) \triangleq \left( \frac{1}{\phi_t(\tau^{(t)}_i)} -\frac{1}{\tau^{(t)}_i} \right)^{-1} \label{g},
\end{gather}
with 
\begin{gather}
        \phi_t(\tau^{(t)}_i) \triangleq mmse\left(x^{(t)}|x^{(t)}+\sqrt{\tau^{(t)}}w^{(t)};
        \{\nu_{\delta_n^{(t)} \rightarrow x_{n}^{(t)}} (x_{n}^{(t)})\}_{n=1}^N
        \right) \label{phi1},
\end{gather}
where  $f(\cdot)$ and $g_t(\cdot)$ respectively represents the transfer functions of the linear estimation module and the MMSE denoiser, the $mmse$ is taken with respect to the prior $\{\nu_{\delta_n^{(t)} \rightarrow x_{n}^{(t)}} (x_{n}^{(t)})\}_{n=1}^N$, and $\{\mathbf{w}^{(t)}\}_{t=1}^{T}$ are IID Gaussian noise series independent of $\{\mathbf{x}^{(t)}\}_{t=1}^{T}$ with zero mean and unit variance. The initialization is $v^{(t)}_0 ={\rm Var}[x^{(t)}]$. The above recursion is depicted in Fig. 4(a). \\
\indent The $mmse$ in $\eqref{phi1}$ involves the messages from the $(t-1)$-th round, i.e.,  $\{\nu_{\delta_n^{(t)} \rightarrow x_{n}^{(t)}} (x_{n}^{(t)})\}_{n=1}^N$, and thus is difficult to analyze. To avoid this difficulty, we first represent the MMSE denoiser in Fig. 4(a) by its equivalent form in Fig 4(b); see the dotted boxes in Fig. 4. Clearly, the inputs of the MMSE denoisers for round $1$ to $t-1$ can be similarly modelled by \eqref{model3}. It is natural to assume that the turbo iteration in each communication round always converges prior to the beginning of the next round. Thus, for round $t$, the input of the round-$t'$ MMSE denoiser can be modelled by 
\begin{equation}
    z_*^{(t')} = x^{(t')} + \sqrt{\tau_*^{(t')}} w_*^{(t')} \text{\quad for \quad} t'=1,...,t-1,
\end{equation}
where $\tau_*^{(t')}$ is the fixed point of $\{ \tau_i^{(t')}\}_i$. Then,
we rewrite $\phi_t(\cdot)$ in \eqref{phi1} as 
\begin{gather}
    \phi_t(\tau^{(t)}_i;\tau^{(t-1)}_*,...,\tau^{(1)}_*) = mmse\left(x^{(t)}|x^{(t)}+\sqrt{\tau^{(t)}}w^{(t)};\mathcal{Z}_*^{(t),pre}\right) \label{phi},\\
    \mathcal{Z}_*^{(t),pre} \triangleq \left\{ x^{(t-1)}+\sqrt{\tau^{(t-1)}_*}w^{(t-1)},...,x^{(1)}+\sqrt{\tau^{(1)}_*}w^{(1)} \right\} \label{z^pre},
\end{gather}
where the $mmse$ is taken with respect to $\{\mathbf{x}^{(t)}\}_{t=1}^T$ modelled in \eqref{model_SE}. Through \eqref{phi}, the transfer function $g_t(\cdot)$ at the $t$-th round in \eqref{g} is dependent on the fixed points in all the past rounds, i.e., $\tau^{(t-1)}_*,...,\tau^{(1)}_*$.  
The next theorem states the monotonicity of the sequence $\{ \tau_i^{(t)}\}_i$ and therefore 
makes sure the existence of the SE fixed point $\tau_*^{(t)}$.
\begin{figure}[!t]
\centering
\subfigure[The block diagram for SE recursion at round $t$]{
\centering
\includegraphics[scale=0.40]{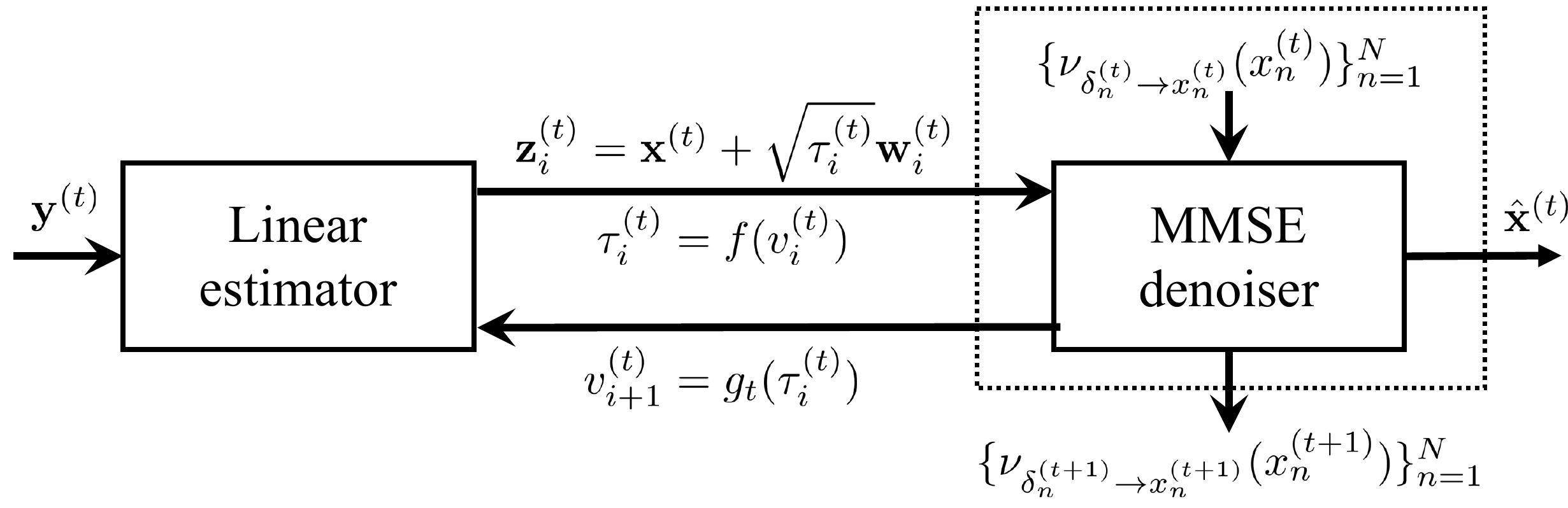}
}%

\subfigure[The block diagram for SE recursion at round $t$]{
\centering
\includegraphics[scale=0.40]{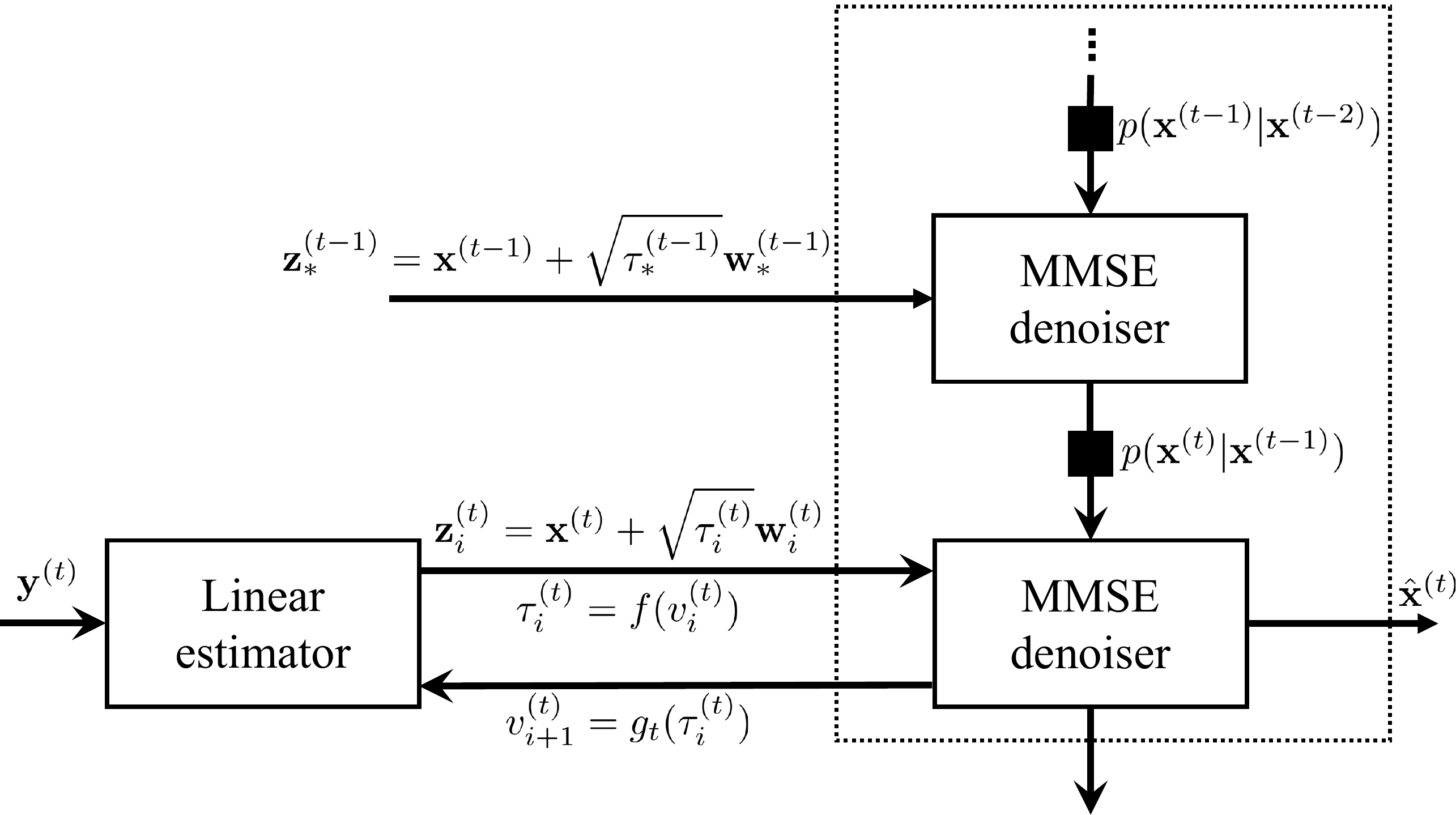}
}%
\centering
\caption{Illustration for SE recursion \eqref{f}\eqref{g}.}
\end{figure}


\newtheorem{myThm}{Theorem}
\begin{myThm} 
By recursion \eqref{f} and \eqref{g}, it holds that
\begin{gather}
    \sigma^2 \leq ... \leq \tau^{(t)}_{i} \leq \tau^{(t)}_{i-1} \leq ...<\tau^{(t)}_1 \label{Th1-1},\\
      0 \leq ... \leq v^{(t)}_{i} \leq v^{(t)}_{i-1}\leq...\leq v^{(t)}_1 \label{Th1-2},
\end{gather}
for $\forall t$, and therefore the recursion convergences to some fixed point.
\end{myThm}
\begin{proof}
See Appendix B.
\end{proof}
We now present the main theorem, which coincides with the intuition that taking the temporal structure into account improves the estimation over time.
\begin{myThm} Assume that $\{\mathbf{x}^{(t)}\}_{t=1}^{T}$ is strictly stationary. Then, for any $i=1,2,...$, we have
\begin{gather}
        \tau^{(t)}_{i} \leq \tau^{(t-1)}_{i} \leq ...<\tau^{(1)}_i \label{Th2-1},\\
        v^{(t)}_{i} \leq v^{(t-1)}_{i}\leq...\leq v^{(1)}_i. \label{Th2-2}
\end{gather}
In particular, the same monotonicity holds for the fixed points $\{\tau_*^{(t)}\}_{t=1}^T$ and $\{v_*^{(t)}\}_{t=1}^T$.
\end{myThm}

\begin{proof}
See Appendix C.
\end{proof}
\subsection{Convergence analysis of FEEL}
Based on Theorem 2, we next analyse the convergence rate of the proposed TSA-GA algorithm for over-the-air FEEL. For simplicity, we only focus on the case $E=1$. Our analysis is based on the assumptions below, which are standard in the stochastic optimization literature\cite{han2015learning}.\\

\noindent \textbf{Assumptions}
\vspace{2ex}
\begin{itemize}
\item [1)] The loss function $\mathcal{L}(\mathbf{u})$ is $c$-stronly convex, i.e., $ \forall \mathbf{u},\mathbf{v} \in \mathbb{R}^d$,
\begin{equation}
    \mathcal{L}(\mathbf{u})-\mathcal{L}(\mathbf{v}) \geq \nabla \mathcal{L}(\mathbf{u})^{T}(\mathbf{v}-\mathbf{u}) +\frac{c}{2}\Vert\mathbf{u}-\mathbf{v}\Vert_2^2.
\end{equation}

\item [2)] The gradient of $\mathcal{L}(\mathbf{u})$ is $L$-Lipschitz, i.e., $ \forall \mathbf{u},\mathbf{v} \in \mathbb{R}^d$,
\begin{equation}
    \Vert \nabla \mathcal{L}(\mathbf{u})- \nabla \mathcal{L}(\mathbf{v}) \Vert \leq L \Vert \mathbf{u}-\mathbf{v} \Vert.
\end{equation}
\item [3)]
The $\ell_2$ norm of the gradient of $\mathcal{L}(\boldsymbol{\theta}^{(t)})$ is bounded over time, i.e., $\forall t$,
\begin{equation}
\Vert \nabla \mathcal{L}(\boldsymbol{\theta}^{(t)}) \Vert \leq G,
\end{equation}
where $G$ is a constant.
\end{itemize}

\begin{myThm} Let learning rate $\eta^{(t)} = \frac{1}{L}, \forall t$ be fixed. After $T$ rounds of communication, the expected learning loss satisfies
\begin{equation}
    \mathbb{E}[\mathcal{L}(\boldsymbol{\theta}^{(T+1)})]-\mathcal{L}(\boldsymbol{\theta}^{*}) \leq \left( \mathbb{E}[\mathcal{L}(\boldsymbol{\theta}^{(1)})]-\mathcal{L}(\boldsymbol{\theta}^{*}) \right) \left (1-\frac{c}{L} \right)^T + \frac{1}{L} \sum_{t=1}^{T} \left (1-\frac{c}{L} \right)^{T-t} \kappa^{(t)} \label{convergence},
\end{equation}
where 
\begin{align}
    \kappa^{(t)} \triangleq  \phi_t \left(\tau^{(t)}_*;\tau^{(t-1)}_*,...,\tau^{(1)}_*\right) + \left[G \rho \left( \frac{(1+\rho)(1-\rho^t)}{1-\rho}+1 \right) \right]^2 \label{kappa},
\end{align}
with $\rho \triangleq \sqrt{\frac{N-k}{N}}$ and $\boldsymbol{\theta}^{*}$ being the optimal model parameter. The expectation is with respect to the randomness of compression and the MAC channel noise.
\end{myThm}

\begin{proof}
See Appendix D.
\end{proof}

We highlight that $\kappa^{(t)}$ in \eqref{kappa} comprises two additive terms that accounts for the deterioration of learning performance incurred by the imperfect message passing recovery at the PS and the sparsification before transmission, respectively. Moreover, given stationary $\{\mathbf{x}^{(t)}\}_{t=1}^{T}$, Theorem 2 states that $\phi_t \left(\tau^{(t)}_*;\tau^{(t-1)}_*,...,\tau^{(1)}_*\right)$ is no greater than $\phi_1 (\tau^{(1)}_*)$. Based on this, we readily obtain that the RHS of \eqref{convergence} is upper bounded by the constant $\frac{1}{c}\bigg{(}\phi_1(\tau^{(1)}_*)+\left(G \rho \left( \frac{1+\rho}{1-\rho}+1 \right) \right)^2\bigg{)}$
as $T \to \infty$.

\section{Real-Time Parameter Learning}
In Section \MakeUppercase{\romannumeral2}, we have assumed the Markovian prior \eqref{support}, \eqref{amplitude} on $\{\mathbf{x}^{(t)}\}_{t=1}^{T}$, where the chain parameter set
$(\lambda,p_{01},p_{10}, \beta, \gamma, \xi) \triangleq \Theta$ remains to be determined in practical implementation of the TSA-GA algorithm; see Line 14 of Algorithm 1. To address this issue, we propose a real-time parameter decision strategy based on the EM principle \cite{dempster1977maximum}. Specifically, at each round $t$, the parameter set $\Theta$ for the next round is learned through the maximum-likelihood statistics using data obtained in the last few rounds, e.g., within a length-$t_0$ time window from $t-t_0+2$ to $t+1$. 
We then formulate the following EM objective:
\begin{align}
    Q(\Theta) ={}& \mathbb{E}\bigg{[}\log p(\mathbf{y}^{(t-t_0+1)}|\mathbf{s}^{(t-t_0+1)} \odot \mathbf{r}^{(i)}) \prod_{n=1}^{N}  \nu_{s_n^{(t-t_0)} \rightarrow s_n^{(t-t_0+1)}} (s_n^{(t-t_0+1)}) \nu_{r_n^{(t-t_0)} \rightarrow r_n^{(t-t_0+1)}}(r_n^{(t-t_0+1)}) \notag \\ & \prod_{i=t-t_0+2}^{t}
    p(\mathbf{y}^{(i)}|\mathbf{s}^{(i)} \otimes \mathbf{r}^{(i)}) 
    \prod_{n=1}^{N} p(s_n^{(i)}|s_n^{(i-1)}) p(r_n^{(i)}|r_n^{(i-1)})  
    \bigg{|}\{\mathbf{y}^{(j)}\}_{j=t-t_0+1}^t\bigg{]}
    \label{EM},
\end{align}
where $\odot$ is the Hadamard product, and the expectation is taken over the hidden variables $\{ \mathbf{s}^{(i)} \}_{i=t-t_0+1}^t $ and $\{ \mathbf{r}^{(i)} \}_{i=t-t_0+1}^t$ given $\{\mathbf{y}^{(j)}\}_{j=t-t_0+1}^t$. We determine each parameter in $\Theta$ by maximizing \eqref{EM}. Due to space limitation, we only consider the case of $p_{01}$ for instance.\\
\indent From \eqref{support}, differentiating \eqref{EM} with respect to $p_{01}$, we obtain
\begin{align}
    \frac{\partial Q}{\partial p_{01}} =
    \sum_{i=t-t_0+2}^{t}
    \sum_{n=1}^{N} {} & \mathbb{E} \left[ \frac{\partial p(s_n^{(i)}|s_n^{(i-1)})} {\partial p_{01}}  
    \bigg{|}\{\mathbf{y}^{(j)}\}_{j=t-t_0+1}^t \right] \\
    = \sum_{i=t-t_0+2}^{t}
    \sum_{n=1}^{N} {} & \mathbb{E} \left[ \frac{1}{p_{01}}(1-s_n^{(i)})s_n^{(i-1)}-\frac{1}{1-p_{01}}s_n^{(i)}s_n^{(i-1)}
    \bigg{|}\{\mathbf{y}^{(j)}\}_{j=t-t_0+1}^t  \right] \\ \notag
    = \sum_{i=t-t_0+2}^{t}
    \sum_{n=1}^{N} {} & \frac{1}{p_{01}} \left( \mathbb{E}\left[s_n^{(i-1)} | \{\mathbf{y}^{(j)}\}_{j=t-t_0+1}^t \right]-\mathbb{E}\left[s_n^{(i)}s_n^{(i-1)} | \{\mathbf{y}^{(j)}\}_{j=t-t_0+1}^t\right] \right) \\ 
     & -\frac{1}{1-p_{01}}\mathbb{E}\left[s_n^{(i)}s_n^{(i-1)}| \{\mathbf{y}^{(j)}\}_{j=t-t_0+1}^t \right]. \label{p_{01}}
\end{align}
Setting \eqref{p_{01}} equal to zero and solving the equation yield the EM update as
\begin{equation}
    p_{01} = 1- \frac{\sum_{i=t-t_0+2}^{t}
    \sum_{n=1}^{N} \mathbb{E}\left[s_n^{(i)}s_n^{(i-1)} | \{\mathbf{y}^{(j)}\}_{j=t-t_0+1}^t\right]}
    {\sum_{i=t-t_0+2}^{t}
    \sum_{n=1}^{N} \mathbb{E}\left[s_n^{(i-1)} | \{\mathbf{y}^{(j)}\}_{j=t-t_0+1}^t \right]}. \label{update}
\end{equation}
We now focus on approximating the posterior expectation in \eqref{update}.
Note that in Section \MakeUppercase{\romannumeral3}, we have already calculated the messages that propagate in the forward direction; see \eqref{delta-s} and \eqref{across1}. To determine the posterior expectation in \eqref{update}, the \emph{backward} message from $p(s_n^{(i)}|s_n^{(i-1)})$ to $s_n^{(i-1)}$ denoted by $\nu_{s_n^{(i-1)}\leftarrow s_n^{(i)}}(s_n^{(i-1)})$ is desired as
\begin{align}
    \nu_{s_n^{(i-1)}\leftarrow s_n^{(i)}}(s_n^{(i-1)}) = 
    \lambda_{s_n^{(i-1)}\leftarrow s_n^{(i)}} \delta(s_n^{(i-1)}-1) +
    (1-\lambda_{s_n^{(i-1)}\leftarrow s_n^{(i)}}) \delta(s_n^{(i-1)}) \label{back_s},
\end{align}
with
\begin{multline}
\lambda_{s_n^{(i-1)}\leftarrow s_n^{(i)}} =\\
   \!\!\!\!\begin{cases}
      \frac{p_{01} 
    (1- \lambda_{\delta_n^{(i)}\rightarrow s_{n}^{(i)}}) (1 - \lambda_{s_n^{(i)}\leftarrow s_n^{(i+1)}} ) + (1-p_{01})\lambda_{\delta_n^{(i)}\rightarrow s_{n}^{(i)}} \lambda_{s_n^{(i)}\leftarrow s_n^{(i+1)}}}
    {(1-p_{01}+p_{01}) (1- \lambda_{\delta_n^{(i)}\rightarrow s_{n}^{(i)}}) (1 - \lambda_{s_n^{(i)}\leftarrow s_n^{(i+1)}} )+ (1-p_{01}+p_{10})\lambda_{\delta_n^{(i)}\rightarrow s_{n}^{(i)}} \lambda_{s_n^{(i)}\leftarrow s_n^{(i+1)}}},\!\! & \!\!\text{if } t-t_0+2 \!\leq \!i\!\leq \!t-1,\\
     \lambda_{\delta_n^{(t)}\rightarrow s_{n}^{(t)}},  \!\!& \!\!\text{if } \!i\!=\!t.\\ 
    \end{cases}
\end{multline}
Combining \eqref{support}, \eqref{delta-s} \eqref{across1} and \eqref{back_s}, we approximate the marginal posterior PDFs as
about $s_n^{t}$ and $s_n^{(t-1)}$ as
\begin{align}\notag
    p(s_n^{(i-1)},s_n^{(i)}|\{\mathbf{y}^{(j)}\}_{j=t-t_0+1}^t ) \approx  {}  &
    \nu_{s_n^{(i)}\leftarrow s_n^{(i+1)}}(s_n^{(i)})
    \nu_{\delta_n^{(i)}\rightarrow s_n^{(i)}}(s_n^{(i)}) \\ \label{p(ss)}
     & \times 
    \nu_{s_n^{(i-2)}\rightarrow s_n^{(i-1)}}(s_n^{(i-1)})
    \nu_{\delta_n^{(i-1)}\rightarrow s_n^{(i-1)}}(s_n^{(i-1)})
    p(s_n^{(i)}|s_n^{(i-1)}) ,\\ \label{p(s)}
    p(s_n^{(i)}|\{\mathbf{y}^{(j)}\}_{j=t-t_0+1}^t )\approx {} &
    \nu_{s_n^{(i)}\leftarrow s_n^{(i+1)}}(s_n^{(i)})
    \nu_{\delta_n^{(i)}\rightarrow s_n^{(i)}}(s_n^{(i)}) 
    \nu_{s_n^{(i-1)}\rightarrow s_n^{(i)}}(s_n^{(i)}).
\end{align}
Then, the posterior expectations in \eqref{update} are with respect to \eqref{p(ss)} and \eqref{p(s)}.
The EM updates for other chain parameters can be derived in a similar way, with the results listed in Algorithm 2. 

\begin{algorithm}[!t] 
\caption{Real-time Parameter Decision for TSA-GA Algorithm} 
\label{alg2} 
\begin{algorithmic}[1] 
\STATE {\% Invoked by Algorithm 1} \\
\IF{$t>t_0$}
\FOR{$i=t, t-1, ..., t-t_0+2$ \%}
\STATE {Compute the message from $s_n^{(i)}$ to $s_n^{(i-1)}$ according to (65)}
\STATE {Compute the message from $r_n^{(i)}$ to $r_n^{(i-1)}$ according to (67)(68)}
\ENDFOR

\STATE {\% Determine the chain parameters for the forward passing next round}
\STATE {$\lambda= \frac{1}{N} \frac{1}{t_0}
                        \sum_{i=t-t_0+1}^t\sum_{n=1}^N \mathbb{E}[s_n^{(i)}|\{\mathbf{y}^{(j)}\}_{j=t-t_0+1}^t]$}
\STATE {$p_{01}= 1- \frac{\sum_{i=t-t_0+2}^{t}
    \sum_{n=1}^{N} \mathbb{E}\left[s_n^{(i)}s_n^{(i-1)} | \{\mathbf{y}^{(j)}\}_{j=t-t_0+1}^t\right]}
    {\sum_{i=t-t_0+2}^{t}
    \sum_{n=1}^{N} \mathbb{E}\left[s_n^{(i-1)} | \{\mathbf{y}^{(j)}\}_{j=t-t_0+1}^t \right]}$}
\STATE {$p_{10}=\frac{\lambda             p_{01}}{1-\lambda}$} 
\STATE {$\beta=\frac{1}{N(t_0-1)}(b+\sqrt{b^2+8N(t_0-1)c} )$ where\\
$b \triangleq \frac{2}{\gamma} \sum_{i=t-t_0+2}^{t}
    \sum_{n=1}^{N} \mathbb{E}[r_n^{(i)} r_n^{(i-1)}|\{\mathbf{y}^{(j)}\}_{j=t-t_0+1}^t]   -   \mathbb{V}[r_n^{(i-1)}|\{\mathbf{y}^{(j)}\}_{j=t-t_0+1}^t] $\\
      $\qquad-   (\mathbb{E}[r_n^{(i-1)}|\{\mathbf{y}^{(j)}\}_{j=t-t_0+1}^t])^2$\\
$c\triangleq \frac{2}{\gamma} \sum_{i=t-t_0+2}^{t}
    \sum_{n=1}^{N} \mathbb{V}[r_n^{(i)}|\{\mathbf{y}^{(j)}\}_{j=t-t_0+1}^t]
    +(\mathbb{E}[r_n^{(i)}|\{\mathbf{y}^{(j)}\}_{j=t-t_0+1}^t])^2$\\
    $\qquad +\mathbb{V}[r_n^{(i-1)}|\{\mathbf{y}^{(j)}\}_{j=t-t_0+1}^t]+(\mathbb{E}[r_n^{(i-1)}|\{\mathbf{y}^{(j)}\}_{j=t-t_0+1}^t])^2 - 2\mathbb{E}[r_n^{(i)} r_n^{(i-1)}|\{\mathbf{y}^{(j)}\}_{j=t-t_0+1}^t]$}
\STATE {$\gamma= \frac{1}{N} \frac{1}{t_0} 
                        \sum_{i=t-t_0+1}^t\sum_{n=1}^N \mathbb{V}[r_n^{(i)}|\{\mathbf{y}^{(j)}\}_{j=t-t_0+1}^t]
                 +(\mathbb{E}[r_n^{(i)}|\{\mathbf{y}^{(j)}\}_{j=t-t_0+1}^t])^2$}
\STATE {$\xi = \frac{(2-\beta) \gamma}{\beta}$}
\ENDIF
\end{algorithmic}
\end{algorithm}

\section{Experimental Results}
In this section, we conduct a series of experiments to test the performance of the proposed TSA-GA algorithm. We consider the image classification task on the MNIST dataset, which consists of 60,000 training and 10,000 test data samples. A single layer neural network with $d=7,850$ parameters is trained with a fixed learning rate $\eta_t=0.01$. In the simulations, the FEEL performance is evaluated by the test accuracy defined as the ratio of the test sample correctly classified to the whole test dataset. By default, $K=25$ devices participate in the FEEL task, each with $K_m=1000$ local data samples drawn from 50,000 MNIST training samples in an IID manner. At each communication round, the devices compute the local gradient using all the $K_m$ data samples. We set the channel noise power $\sigma_e^2=1$.\\
\indent 
We describe the initialization of the chain parameters as follows. The initial prior sparsity $\lambda$ is set to $k/N$, which is exactly the case when the local gradients share the same sparsity pattern. The support transition probability $p_{01}$ and the amplitude forgetting coefficient $\alpha$ are preferred to be small so that the aggregation estimator is obliged to make full use of the temporal structure to assist the recovery. Empirically, setting initial $p_{01} = 0.005$ and $\alpha = 0.005$ promotes the acquirement of the historical knowledge and accelerates the  training at the early stages. The initial variance $\gamma$ is estimated from the observation $\mathbf{y}$ via $\gamma = \frac{\Vert{\mathbf{y}^{(1)}}\Vert_2^2}{N\lambda}$. The remaining chain parameters $p_{10}$ and $\xi$ are initialized according to Line 10 and Line 13 of Algorithm 2. We set $\epsilon=10^{-7}$. In Algorithm 1, the EM update is invoked after the first 10 communication rounds. This is because the estimation of $\mathbf{x}^{(t)}$ at the initial stages is not accurate enough for the EM procedure to provide a reliable parameter update. After the first 10 rounds, we find that a window length of $t_{0}=5$ works well for Algorithm 2
to track the chain parameters throughout the training. We set $I=25$ to ensure convergence.\\
\indent We employ the following benchmarks for performance comparison. 
\begin{enumerate}
    \item Error-free channel: Suppose that the transmission error, including the error caused by compression/decompression and channel noise, is free. Thus, at the $t$-th communication round, the global model is updated via $    \boldsymbol{\theta}^{(t+1)}=\boldsymbol{\theta}^{(t)}-\eta^{(t)} \frac{1}{K} \sum_{m=1}^M K_m {\mathbf{g}_m^{sp^{(t)}}}$.
    \item TSA-GA without support correlation: This case adopts the proposed TSA-GA algorithm by removing the Markov chain of the support. That is, $p(s_n^{(j)}|s_n^{(j-1)})$ in \eqref{post} is replaced with $p(s_n^{(j)})=\lambda\delta(s_n^{(j)}-1)+(1-\lambda)\delta(s_n^{(j)}).$
    \item TSA-GA without support correlation: Similarly, $p(r_n^{(j)}|r_n^{(j-1)})$ in \eqref{post} is replaced with $\mathcal{N}\left(r_n^{(j)};0,\gamma\right)$. Thus only amplitude correlation is incorporated. 
    
    \item A-DSGD: Aggregating without correlation awareness \cite{amiri2020machine},
    where the AMP algorithm\cite{donoho2009message} is adopted for the reconstruction of $\mathbf{x}^{(t)}$ frame by frame. A pseudo-random IID Gaussian matrix with normalized columns is used for each $\mathbf{A}^{(t)}$.
\end{enumerate}

\indent In Fig. 5, we plot the test accuracy versus the communication round for $\bar{P}=500, s/N=0.1$, $k/N=0.05$ for A-DSGD and $k/N=0.2$ otherwise. 
Both the scenarios with and without multiple steps of gradient descent are investigated. In both scenarios, the proposed training approach outperforms other counterparts in terms of convergence speed. Compared with that in Fig. 5(a), the convergence in Fig. 5(b) is faster due to the use of multiple steps of gradient descent in Fig. 5(b). In Fig. 5, the approach \cite{amiri2020machine} that ignores the historical information performs the worst, while the test accuracy of our TSA-GA algorithm approaches that of the ideal error-free scenario. This substantial performance gain is attributed to the reconstruction refinement by taking temporal information into account. Moreover, partial exploitation of the temporal structure of the gradient aggregation (i.e. w/o support or amplitude transition) degrades the performance of TSA-GA. 

\begin{figure*}[!t]
\centering
\subfigure[$E=1$]{
\begin{minipage}[t]{0.5\textwidth}
\centering
\includegraphics[scale=0.55]{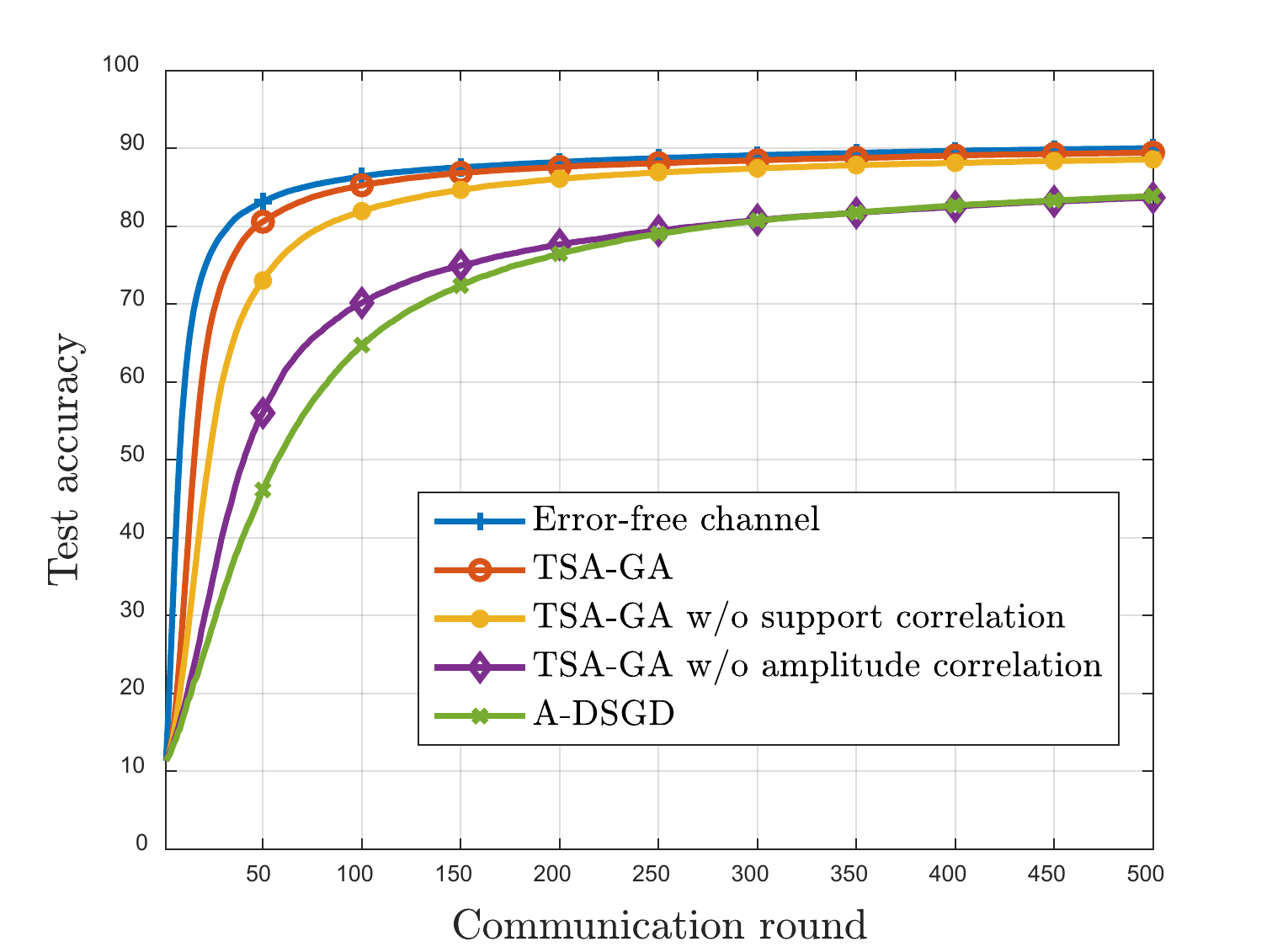}
\end{minipage}%
}%
\subfigure[$E=5$]{
\begin{minipage}[t]{0.5\textwidth}
\centering
\includegraphics[scale=0.55]{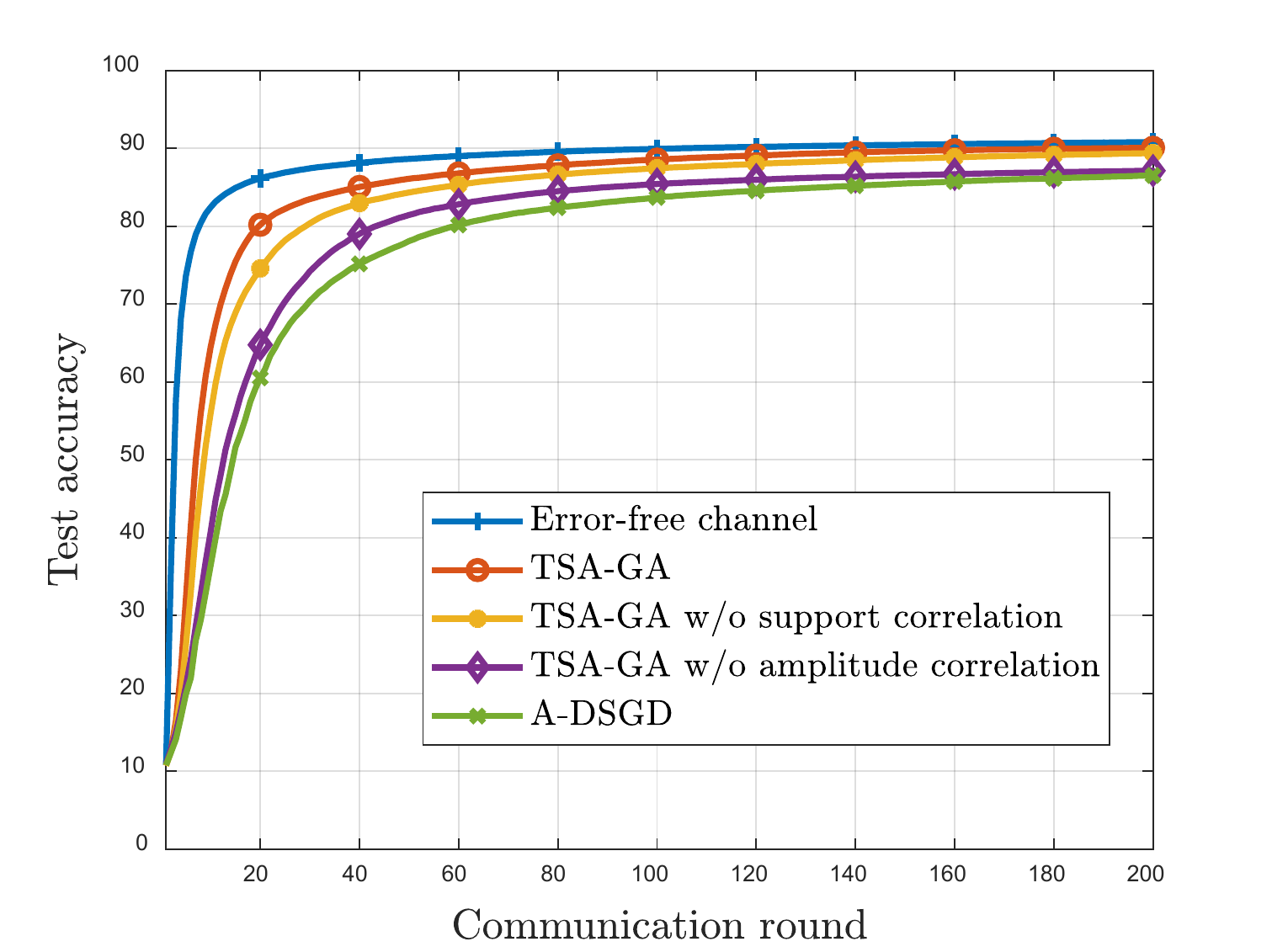}
\end{minipage}%
}%
\centering
\caption{Test accuracy of different algorithms versus communication round, $s/N=0.1$, $\bar{P}=500$, $M=25$, $K_m=1000$, $k/N=0.05$ for A-DSGD and $k/N=0.2$ for others. }
\end{figure*}

\indent In Fig. 6, we investigate the test accuracy for different compression ratios against the communication round and the total number of symbols, respectively. 
Again, both Fig. 6(a) and (b) show that the proposed algorithm achieves considerable performance improvement over the baseline A-DSGD in term of both the convergence rate and final accuracy under various ratios of compression. In addition, it is remarkable that the TSA-GA algorithm works well even when 25 times of compression is used, while the training performance of A-DSGD deteriorates sharply as the compression becomes more aggressive. 
Note that in Fig. 6(b), the convergence rate of TSA-GA accelerates as the $s/N$ decreases from $0.5$ to $0.1$ and maintains almost the same from $0.1$ to $0.04$. The reason is that for $s/N<0.1$, more aggressive compression leads to much more severe degradation of the reconstruction performance at the PS and therefore more rounds of communication are required to achieve the same accuracy.

\begin{figure*}[!t]
\subfigure[Test accuracy versus communication round]{
\begin{minipage}[t]{0.50\textwidth}
\centering
\includegraphics[scale=0.55]{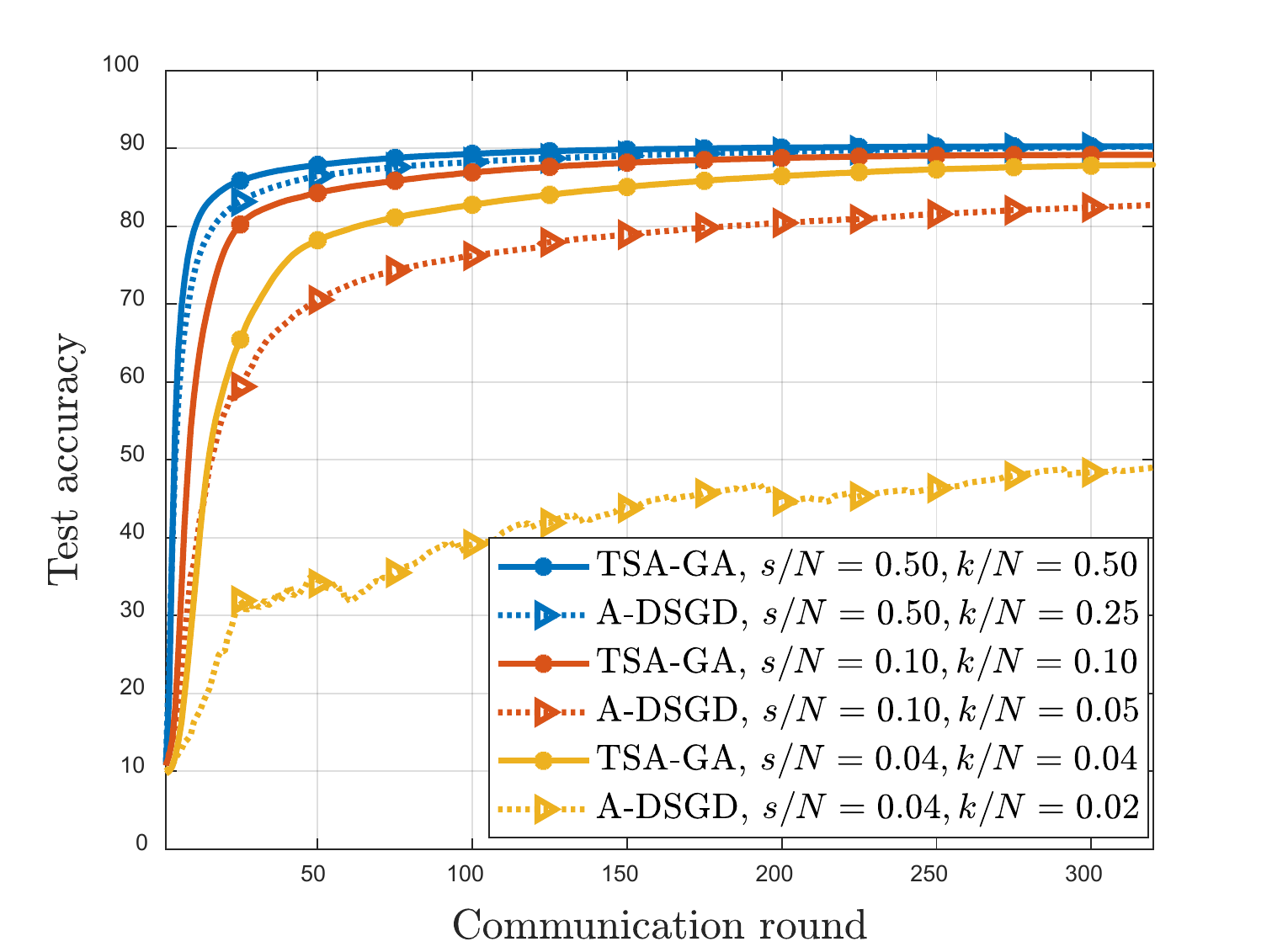}
\end{minipage}%
}%
\subfigure[Test accuracy versus total number of symbols]{
\begin{minipage}[t]{0.50\textwidth}
\centering
\includegraphics[scale=0.55]{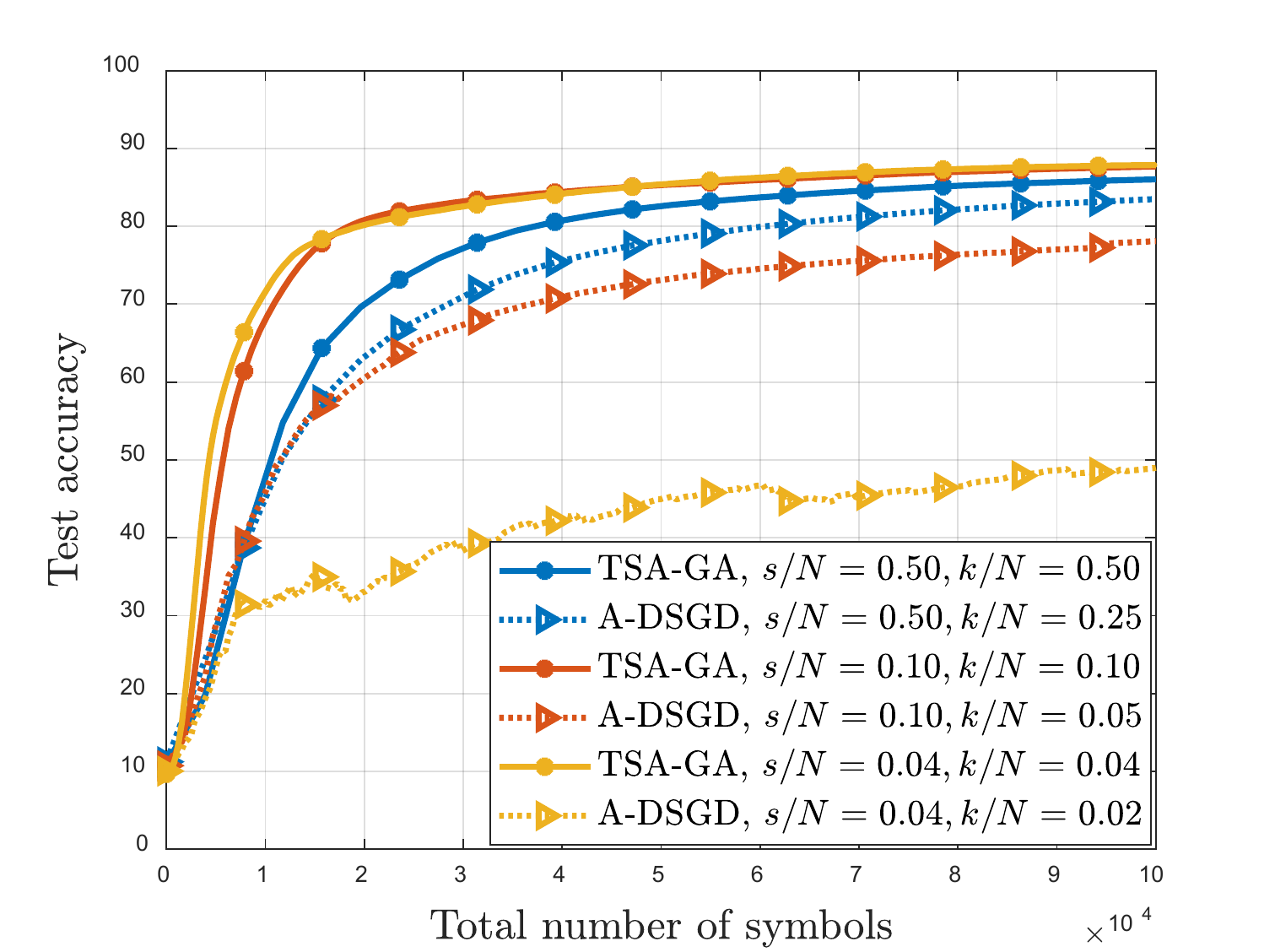}
\end{minipage}%
}%
\centering
\caption{Test accuracy for different compression ratios, $M =25$, $K_m=1000$, $E=5$, $k=s/2$ for A-DSGD and $k=s$ otherwise.}
\end{figure*}

\indent In Fig. 7, we investigate the impact of transmit power $\bar{P}$ on the performance with and without exploiting the temporal structure. We again observe substantial performance improvement by exploiting the inherent temporal structure of the gradient aggregation series. We highlight that as $\bar{P}$ decreases, the performance degradation of the TSA-GA scheme is much slower compared with that of the A-DSGD. This indicates that the proposed scheme is more robust to the channel noise compared with the A-DSGD.  

\indent Finally, we consider the FEEL performance with non-IID data distribution among the devices in Fig. 8, where each device is constrained to select data samples from only $\chi$ classes.
For each device $m$, the local data are selected as follows. Determine randomly which $\chi$ classes the local data come from at first, and then choose data samples uniformly within the given $\chi$ classes. 
From Fig. 8, we observe that under this non-IID data distribution, the FEEL performance degrades compared with the IID case. Decreasing $\chi$ from 5 to 2 further worsens the learning performance as expected. Besides, Fig. 9 demonstrates the superiority of TSA-GA again by noticing that the gap towards the error-free one is much smaller than that of A-DSGD. 



\begin{figure*}[!t]
\begin{minipage}[t]{0.48\textwidth}
\centering
\includegraphics[scale=0.55]{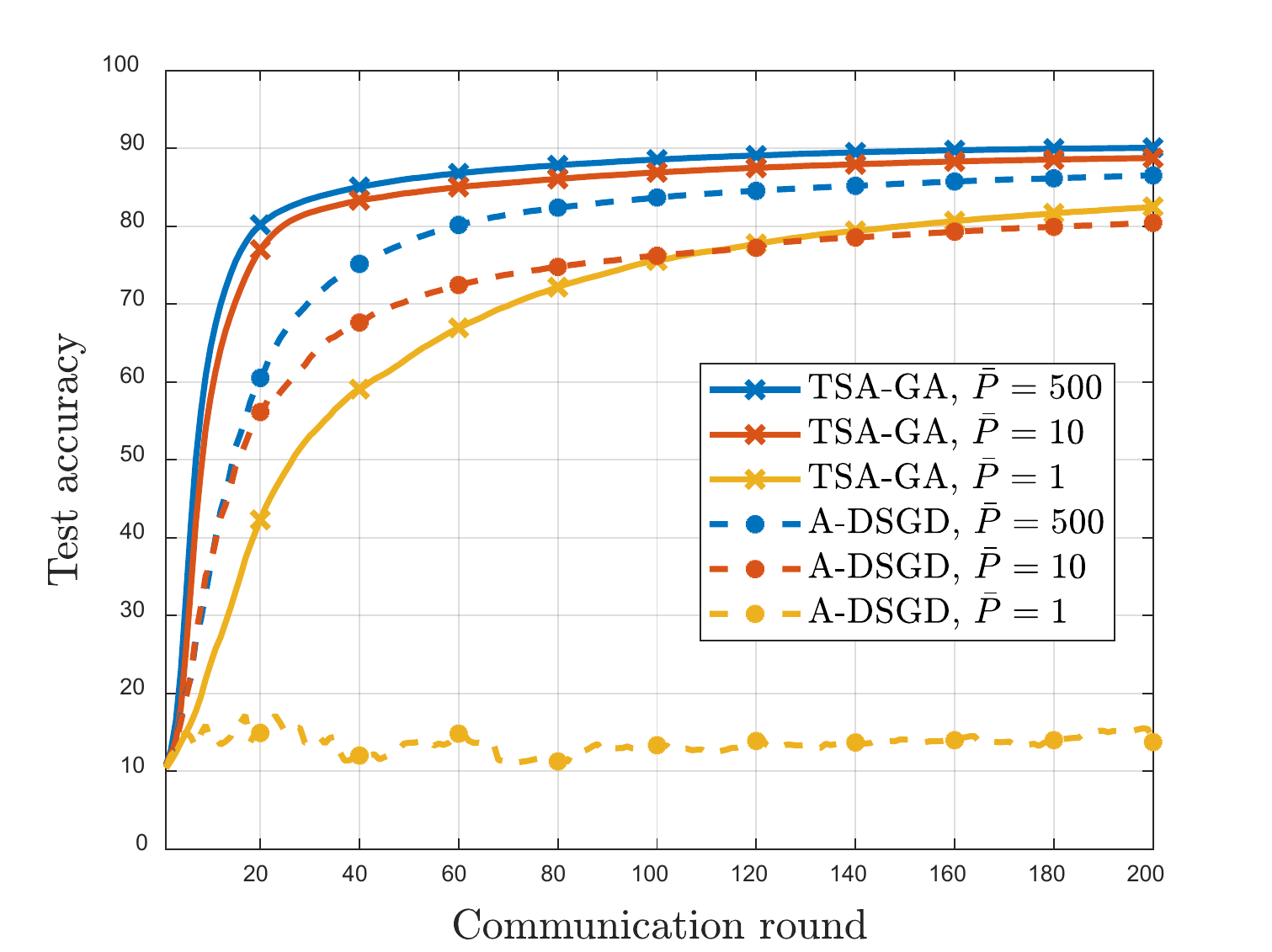}
\caption{Test accuracy versus communication round for different power. $s/N=0.1$, $M =25$, $K_m=1000$, $E=5$, $k=s/2$ for A-DSGD and $k=s$ otherwise.}
\end{minipage}%
\centering
\quad
\begin{minipage}[t]{0.48\textwidth}
\centering
\includegraphics[scale=0.55]{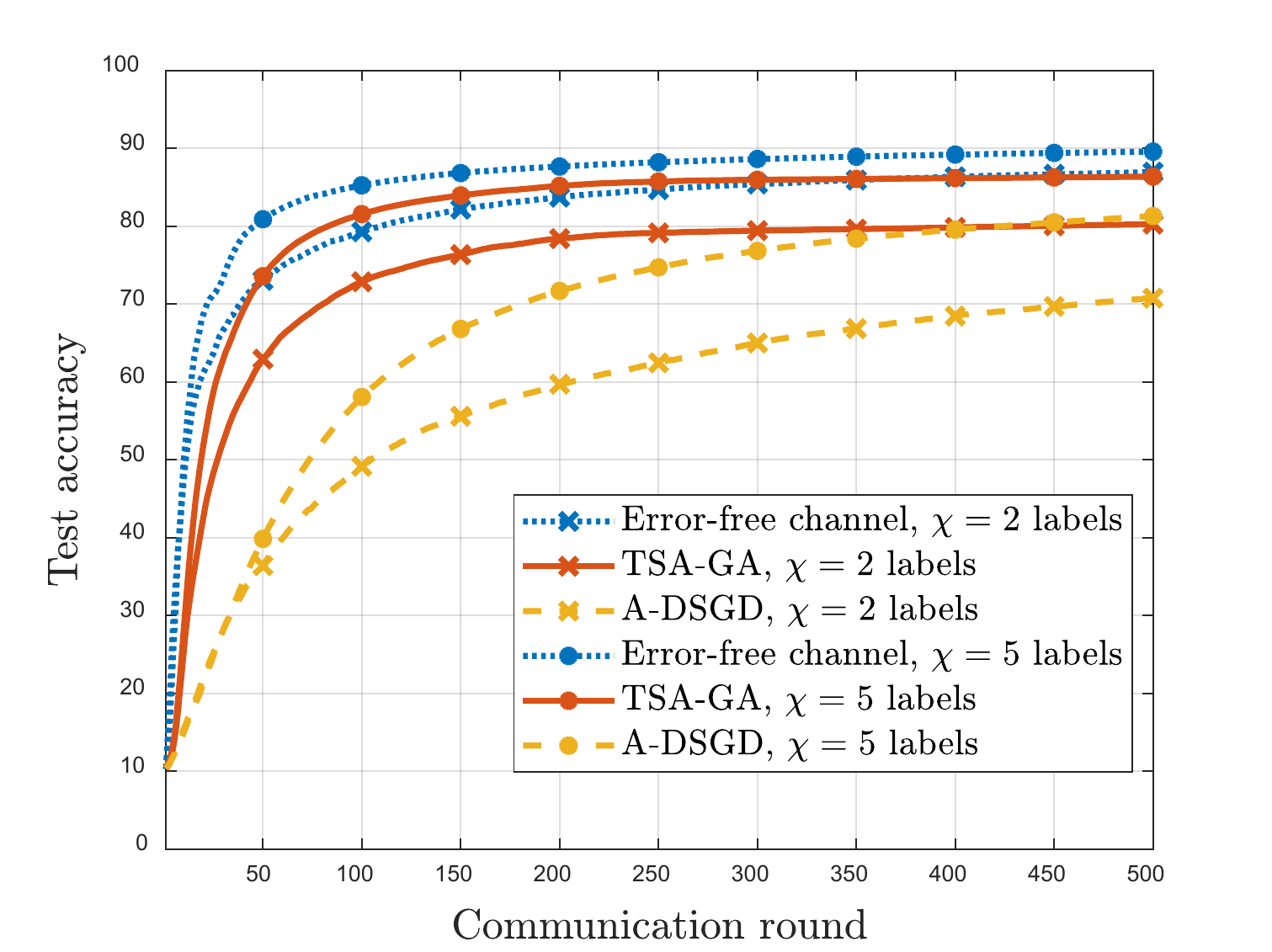}
\caption{Test accuracy versus communication round with non-IID data distribution. $s/N=0.1$, $k=s/2$, $\bar{P}=500$, $M=25$, $K_m=1000$, $E=1$.}
\end{minipage}%
\centering
\end{figure*}

\section{Conclusion}
In this paper, we studied over-the-air model aggregation in the FEEL system. We introduced a Markovian probability model to characterize the temporal structure of the gradient aggregation series. Based on the Markovian model, we developed a turbo message passing algorithm, termed TSA-GA, to efficiently recover the desired gradient aggregation from all the past noisy observations at the PS. We further established the SE analysis to characterize the behaviour of the proposed TSA-GA algorithm. Based on the SE analysis, we established an explicit bound of the expected loss reduction under certain standard regularity conditions. For practical implementation, we developed an EM strategy to determine the unknown parameters in the Markovian model. We showed that the proposed TSA-GA scheme significantly outperforms its counterpart schemes, and that much more aggressive compression of local updates can be achieved by taking the gradient temporal structure into account.

\appendices

\section{Some Useful Lemmas}
 \newtheorem{mylemma}{Lemma}




\begin{mylemma} 
Suppose that $\{\mathbf{x}^{(t)}\}_{t=1}^T$ share the same probability model as $\{\mathbf{x}^{(t)}\}$ in \eqref{model_SE}. 

With $ 0 \leq \tau^{(t)} \leq \tilde{\tau}^{t},0 \leq \tau^{(t-1)} \leq \tilde{\tau}^{(t-1)},...,0 \leq \tau^{(1)} \leq \tilde{\tau}^{(1)}$,
we have
\begin{multline}
         mmse(x^{(t)}|x^{(t)}+\sqrt{\tau^{(t)}}w^{(t)};x^{(t-1)}+\sqrt{\tau^{(t-1)}}w^{(t-1)},...,x^{(1)}+\sqrt{\tau^{(1)}}w^{(1)}) \\
        \leq mmse(x^{(t)}|x^{(t)}+\sqrt{\tilde{\tau}^{(t)}}w^{(t)};x^{(t-1)}+\sqrt{\tilde{\tau}^{(t-1)}}w^{(t-1)},...,x^{(1)}+\sqrt{\tilde{\tau}^{(1)}}w^{(1)}). \label{lemma3}
\end{multline}
\end{mylemma}

\begin{proof}
Since a Gaussian random variable is infiniely divisible, we construct a cascaded AWGN observation for $ {x}^{(1)}$ as follows
\begin{align}
 {z}^{(1)} & =  {x}^{(1)} + \sqrt{\tau^{(1)}} {w}^{(1)} \label{o1}, \\
 {\tilde{z}}^{(1)} & =  {z}^{(1)} + \sqrt{\tilde{\tau}^{(1)}-\tau^{(1)}} {\tilde{w}}^{(1)} \label{o2},
\end{align}
 where $ {x}^{(1)}$,$ {w}^{(1)} \sim \mathcal{N}(0,1)$, $ {\tilde{w}}^{(1)} \sim \mathcal{N}(0,1)$ are independent. Based on this, we obtain
 \begin{align} \notag
    & mmse( {x}^{(t)}| {x}^{(t)}+\sqrt{\tau^{(t)}} {w}^{(t)}; {x}^{(t-1)}+\sqrt{\tau^{(t-1)}} {w}^{(t-1)},..., {x}^{(1)}+\sqrt{\tilde{\tau}^{(1)}} {w}^{(1)}) \\ \notag
 \geq {} & mmse( {x}^{(t)}| {x}^{(t)}+\sqrt{\tau^{(t)}} {w}^{(t)}; {x}^{(t-1)}+\sqrt{\tau^{(t-1)}} {w}^{(t-1)},..., {x}^{(1)}+\sqrt{\tilde{\tau}^{(1)}} {w}^{(1)}, \sqrt{\tilde{\tau}^{(1)}-\tau^{(1)}} {\tilde{w}}^{(1)} ) \\ \notag
 = {} & mmse( {x}^{(t)}| {x}^{(t)}+\sqrt{\tau^{(t)}} {w}^{(t)}; {x}^{(t-1)}+\sqrt{\tau^{(t-1)}} {w}^{(t-1)},..., {x}^{(1)}+\sqrt{\tau}^{(1)} {w}^{(1)}, \sqrt{\tilde{\tau}^{(1)}-\tau^{(1)}} {\tilde{w}}^{(1)} )\\
 = {} & mmse( {x}^{(t)}| {x}^{(t)}+\sqrt{\tau^{(t)}} {w}^{(t)}; {x}^{(t-1)}+\sqrt{\tau^{(t-1)}} {w}^{(t-1)},..., {x}^{(1)}+\sqrt{\tau^{(1)}} {w}^{(1)}),
\end{align}
where the first step follows by adding extra condition $\sqrt{\tilde{\tau}^{(1)}-\tau^{(1)}} {\tilde{w}}^{(1)}$, the second step follows from the fact that knowing $ {\tilde{z}}^{(1)}$ and $\sqrt{\tilde{\tau}^{(1)}-\tau^{(1)}} {\tilde{w}}^{(1)}$ is equivalen to knowing $ {z}^{(1)}$ and $\sqrt{\tilde{\tau}^{(1)}-\tau^{(1)}} {\tilde{w}}^{(1)}$ due to the cascaded observation models \eqref{o1} and \eqref{o2}, and the third step follows from the conditional independence of $ {x}^{(t)}$ and $\sqrt{\tilde{\tau}^{(1)}-\tau^{(1)}} {\tilde{w}}^{(1)}$ given $ {z}^{(1)}$. Finally, we justify \eqref{lemma3} by similar constructions in \eqref{o1} and \eqref{o2} for other $x^{(t')}$, $t'=2,...,t$.
\end{proof}

\begin{mylemma} The transfer functions $f(\cdot)$ and $g_t(\cdot)$
in \eqref{f} and \eqref{g} are monotonically increasing.
\end{mylemma}


\begin{proof}
We basically follow the procedure in \cite{ma2015performance}. The monotonicity of $f(\cdot)$ is evident. As for $g_t(\cdot)$, we consider the monotonicity of $(g_t(\cdot))^{-1} \triangleq \psi_t(\cdot)$ instead. 
The derivative of $\psi_t(\tau^{(t)}_i)$ can be expressed as
\begin{align}
    \hspace{-0.25cm} \frac{d\psi_t}{d\tau^{(t)}_i} \!\!
    = \!\!-\frac{1}{\phi_t(\tau^{(t)}_i)^2} \frac{d\phi_t(\tau^{(t)}_i)}{d\tau^{(t)}_i} \!\!+\!\! \frac{1}{(\tau^{(t)}_i)^2} 
     \!\!\overset{(a)}{=} \!\!\frac{-\mathbb{E} \left[ M_t(\tau^{(t)}_i;\mathcal{Z}_*^{(t),pre})^2 \right] \!\!+\!\!  \left\{\mathbb{E} \left[M_t(\tau^{(t)}_i;\mathcal{Z}_*^{(t),pre})\right] \right\}^2 } { (\tau_i^{(t)})^2 \phi_t(\tau^{(t)}_i)^2
     } \!\!\leq \!\!0 ,
\end{align}
where $M_t(\tau;\mathcal{Z}) \triangleq \mathbb{E}\left[ \left[ x^{(t)}- \mathbb{E}(x^{(t)}|z^{(t)};\mathcal{Z})\right]^2 \big{|}z^{(t)};\mathcal{Z}\right]$, and step (a) is due to \cite[Corolary 2]{guo2011estimation}.
Therefore, $\psi_t(\cdot)$ is monotonically decreasing and thus $g_t(\cdot)$ is monotonically increasing.
\end{proof}

\section{Proof of Theorem 1}
We prove by induction. We first show that, $v^{(t)}_2 \leq v^{(t)}_1$ for any fixed $t$. 
From \eqref{phi}, we have
\begin{align}
    \phi_t(\tau^{(t)}_i) \overset{(a)}
    \leq mmse(x^{(t)}|x^{(t)}+\sqrt{\tau^{(t)}}w^{(t)})
    \overset{(b)} 
    \leq \frac{\tau^{(t)} v^{(t)}_1}{\tau^{(t)}+v^{(t)}_1},
\end{align}
where step $(a)$ follows by dropping extra condition $\mathcal{Z}_*^{(t),pre}$, and step $(b)$ follows from \cite[Proposition 15]{guo2011estimation}.
Thus, 
\begin{align}
    v^{(t)}_2  = \left( \frac{1}{\phi_t(\tau^{(t)}_1)} - \frac{1}{\tau^{(t)}_1}\right)^{-1} 
     \leq \left( \frac{\tau^{(t)}+v^{(t)}_1}{\tau^{(t)} v^{(t)}_1}- \frac{1}{\tau^{(t)}_1}\right)^{-1}
    = v^{(t)}_1.
\end{align}
Now, suppose that \eqref{Th1-1} and \eqref{Th1-2} hold for $i=1,...,I$, Then for $i=I+1$, from \eqref{f} and \eqref{g}, 
\begin{equation}
    v_{I+1}^{(t)} = g_t(\tau_I^{(t)}) = g_t(f(v_I^{(t)})) \leq g_t(f(v_{I-1}^{(t)})) = v_{I}^{(t)},
\end{equation}
where the inequality is due to the induction assuption and the monotonicity of $f(\cdot)$ and $g_t(\cdot)$ in Lemma 2. Moreover, $\tau^{(t)}_{i} \geq \sigma^2$ and $v^{(t)}_{i} \geq 0$ is evident from \eqref{f}, \eqref{g} and \cite[Proposition 15]{guo2011estimation}, which completes the proof.


\section{Proof of Theorem 2}
We prove by nested induction. The outer induction is with respect to round index $t$ and the inner induction is with respect to turbo iteration index $i$. We devide the proof into 3 parts. \\
\indent \emph{Part 1:} To begin with, we show that \eqref{Th2-1} and \eqref{Th2-2} hold for $t=2$ and $i=1$:
\begin{align}
    \phi_{t}(\tau^{(t)}_i) \notag
    & \overset{(a)}{=} mmse\left( {x}^{(t)} \bigg{|} {x}^{(t)}+\sqrt{\tau^{(t)}_i} {w}^{(t)}; \left\{{x}^{(j)}+\sqrt{\tau^{(j)}_*} {w}^{(j)} \right\}_{j=1}^{t-1} \right) \\ \notag
    & \overset{(b)}{\leq}     mmse\left( {x}^{(t)} \bigg{|} {x}^{(t)}+\sqrt{\tau^{(t)}_i} {w}^{(t)}; \left\{{x}^{(j)}+\sqrt{\tau^{(j)}_*} {w}^{(j)} \right\}_{j=2}^{t-1} \right) \\ \notag
    & \overset{(c)}{\leq}     mmse\left( {x}^{(t)} \bigg{|} {x}^{(t)}+\sqrt{\tau^{(t)}_i} {w}^{(t)}; \left\{{x}^{(j)}+\sqrt{\tau^{(j-1)}_*} {w}^{(j)} \right\}_{j=2}^{t-1} \right) \\ 
    & \overset{(d)}{\leq}     mmse\left( {x}^{(t-1)} \bigg{|} {x}^{(t-1)}+\sqrt{\tau^{(t)}_i} {w}^{(t-1)}; \left\{{x}^{(j-1)}+\sqrt{\tau^{(j-1)}_*} {w}^{(j-1)} \right\}_{j=2}^{t-1} \right)\notag \\ & = \phi_{t-1}(\tau^{(t)}_i) \label{phi_mono},
\end{align}
\noindent 
where step (a) follows from the definition in \eqref{phi}, step (b) follows by dropping the condition indexed by $j=1$, step (c) holds because $\{{x}^{(j)}+\sqrt{\tau^{(j)}_*} {w}^{(j)} \}_{j=2}^{t-1}$ is in fact empty for $t=2$, and step (d) follows from the assumption that $\{ {x}^{(t)}\}_{t=1}^{T}$ is strictly stationary. Then,
\begin{align}
    \hspace{-2mm} v_{i+1}^{(t)}\!\overset{(a)}{=} \!\left( \frac{1}{\phi_{(t)}(\tau^{(t)}_i)}\!-\!\frac{1}{\tau^{(t)}_i}\right)^{-1}
     \!\! \!\overset{(b)}{\leq} \!\left( \frac{1}{\phi_{t-1}(\tau^{(t)}_i)}\!-\!\frac{1}{\tau^{(t)}_i}\right)^{-1}
    \!\!\! \overset{(c)}{=}\! g_{t-1}(\tau^{(t)}_i)
    \!\overset{(d)}{\leq} g_{t-1}(\tau^{(t-1)}_i) \!\overset{(e)}{=}\! v_{i+1}^{(t-1)} \label{v_mono},
\end{align}
 where step (a) follows from the recursion \eqref{g}, step (b) follows from \eqref{phi_mono}, step (c) follows from the definition in \eqref{g}, step (d) is due to the monotonicity of $g_{t-1}(\cdot)$ described in Lemma 2 and the fact that $\tau^{(t)}_1$ are identical for strictly stationary $\{ {x}^{(t)}\}_{t=1}^{T}$, and step (e) follows from \eqref{g}.
From $f(\cdot)$ in \eqref{f},
\begin{equation}
    \tau_{i+1}^{(t)} = f(v_{i+1}^{(t)}) \leq f(v_{i+1}^{(t-1)}) =\tau_{i+1}^{(t-1)} \label{tau_mono},
\end{equation}
where the inequality follows from $\eqref{v_mono}$ and the monotonicity of $f(\cdot)$.\\
\indent \emph{Part 2:} We now consider the inner induction with respect to $i$ for fixed $t=2$. Suppose that \eqref{Th2-1} and \eqref{Th2-2} hold for $i=1,...,I$. We need to prove \eqref{Th2-1} and \eqref{Th2-2} hold for $i=I+1$ and $t=2$. We readily see that all the steps of \eqref{phi_mono}-\eqref{tau_mono} hold straightforwardly for $i=I+1$, except for step (d) in \eqref{v_mono}. Yet, step (d) in \eqref{v_mono} holds for $i=I+1$ due to 
the induction assumption and the monotonicity of $g_{t-1}(\cdot)$.\\ 
\indent \emph{Part 3:} So far, we have shown that \eqref{Th2-1} and \eqref{Th2-2} hold for $t=2$ and $\forall i$. We now consider the outer induction with respect to $t$. Suppose that \eqref{Th2-1} and \eqref{Th2-2} hold for $t=2,..,t'$ and $\forall i$. For $t=t'+1$ and $i=1$, we readily see that all the steps of \eqref{phi_mono}-\eqref{tau_mono} hold straightforwardly for $i=1$, except for step (b) in \eqref{phi_mono}. Yet, step (b) in \eqref{phi_mono} holds due to the induction assumption and Lemma 2. Then, the case for $t=t'+1$ and any $i$ can be proved by induction with respect to $i$ in a similar way. Combining the three parts we complete the proof.

\section{Proof of Theorem 3}
We basically follow the procedure used in \cite{liu2020reconfigurable,friedlander2012hybrid}.
From \cite[Theorem 2.2]{friedlander2012hybrid}, we have 
\begin{align}
            \mathbb{E}[\mathcal{L}(\boldsymbol{\theta}^{(T+1)})]-\mathcal{L}(\boldsymbol{\theta}^{*}) \leq {} & \left( \mathbb{E}[\mathcal{L}(\boldsymbol{\theta}^{(1)})]-\mathcal{L}(\boldsymbol{\theta}^{*}) \right) \left (1-\frac{c}{L} \right)^T \notag \\
           & + \frac{1}{2L} \sum_{t=1}^{T}  
        \left (1-\frac{c}{L} \right)^{T-t} \mathbb{E}\left[\bigg{\Vert} \hat{\mathbf{x}}^{(t)}-\frac{1}{K}\sum_{m=1}^{M} K_m \mathbf{g}_m^{(t)} \bigg{\Vert}_2^2 \right]. \label{100-1}
\end{align}
We upper-bound the expected gradient error at each round $t$ by
\begin{align}
     &\mathbb{E}\left[\bigg{\Vert} \hat{\mathbf{x}}^{(t)}-\frac{1}{K}\sum_{m=1}^{M} K_m \mathbf{g}_m^{(t)} \bigg{\Vert}_2^2 \right] \notag \\ ={} & \mathbb{E}\left[\bigg{\Vert} \left(\hat{\mathbf{x}}^{(t)}-\frac{1}{K}\sum_{m=1}^{M} K_m {\mathbf{g}_m^{sp^{(t)}}} \right) + \left(\frac{1}{K}\sum_{m=1}^{M} K_m{\mathbf{g}_m^{sp^{(t)}}} - \frac{1}{K}\sum_{m=1}^{M}K_m \mathbf{g}_m^{(t)} \right)
    \bigg{\Vert}_2^2 \right] \notag \\
    \leq {} & 2\mathbb{E}\left[\bigg{\Vert} \left(\hat{\mathbf{x}}^{(t)}-\frac{1}{K}\sum_{m=1}^{M} K_m{\mathbf{g}_m^{sp^{(t)}}} \right) \bigg{\Vert}_2^2 \right] + 2\bigg{\Vert} \frac{1}{K}\sum_{m=1}^{M}K_m {\mathbf{g}_m^{sp^{(t)}}} - \frac{1}{K}\sum_{m=1}^{M}K_m \mathbf{g}_m^{(t)} 
    \bigg{\Vert}_2^2 \label{84},
\end{align}
where the last inequality is due to the Cauchy-Schwartz inequlity. Note that the first term is simply $\phi_t \left(\tau^{(t)}_*;\tau^{(t-1)}_*,...,\tau^{(1)}_*\right)$ in \eqref{phi} from the SE analysis. As for the second term,
\begin{align}
    \notag & \bigg{\Vert} \frac{1}{K}\sum_{m=1}^{M}K_m {\mathbf{g}_m^{sp^{(t)}}} - \frac{1}{K}\sum_{m=1}^{M}K_m \mathbf{g}_m^{(t)}
    \bigg{\Vert}_2^2 \\ \overset{(a)}{=} {} &
    \bigg{\Vert} \frac{1}{K}\sum_{m=1}^{M} K_m\left(\boldsymbol{\Delta}_m^{(t)}-\boldsymbol{\Delta}_m^{(t+1)}\right) \bigg{\Vert}_2^2  \notag \\  \overset{(b)}{\leq} {} &
    \left( \frac{1}{K}\sum_{m=1}^{M} K_m\left( \Vert \boldsymbol{\Delta}_m^{(t)}\Vert_2+ \Vert \boldsymbol{\Delta}_m^{(t+1)}\Vert_2\right) \right)^2
    \notag \\ \overset{(c)}{\leq} {} &
    \left( \frac{1}{K}\sum_{m=1}^{M} \sum_{i=1}^{t-1} K_m\rho^{t-i} \Vert \mathbf{g}_m^{(i)} \Vert_2 + \frac{1}{K}\sum_{m=1}^{M}\sum_{i=1}^{t}K_m \rho^{t+1-i} \Vert \mathbf{g}_m^{(i)} \Vert_2 \right)^2 \notag \\ \overset{(d)}{\leq} {} & 
    \left[G \rho \left( \frac{(1+\rho)(1-\rho^t)}{1-\rho}+1 \right) \right]^2 \label{106},
\end{align}
where step (a) is due to \eqref{ec}-\eqref{Delta}, step (b) is due to the triangle inequality, 
step (c) is due to \cite[(49)]{amiri2020machine}, and step (d) is due to Assumption 3). Combining \eqref{100-1}-\eqref{106}, we finally obtain the desired result in \eqref{convergence}.

\bibliographystyle{IEEEtran}
\bibliography{citationlist}

\end{document}